\documentclass[a4paper,UKenglish,cleveref,autoref,thm-restate]{lipics-v2021}

\usepackage{todonotes}
\usepackage{microtype}
\usepackage{nicefrac}

\usepackage{MnSymbol}

\usepackage[sort]{cite}
\usepackage{xspace}
\usepackage{csquotes}
\usepackage{nicefrac}
\usepackage{mathtools}
\usepackage{tikz}
\usepackage{tikz-cd}
\usetikzlibrary{calc}
\usepackage{complexity}

\graphicspath{{./}{./figures/}}
\DeclareGraphicsExtensions{.pdf,.PDF}

\usepackage{subcaption}
\usepackage[normalem]{ulem}

\crefname{observation}{Observation}{Observations}

\newcommand{\diam}{\ensuremath{\operatorname{diam}}}
\newcommand{\dist}{\ensuremath{\operatorname{dist}}}

\newcommand{\las}{\operatorname{ac}}

\title{Flip Distance of Non-Crossing Spanning Trees: \\ NP-Hardness and Improved Bounds}

\authorrunning{H.~Bjerkevik, J.~Dorfer, L.~Kleist, T.~Ueckerdt, and B.~Vogtenhuber}

\author{H\aa vard Bakke Bjerkevik}{University at Albany}{hbjerkevik@albany.edu}{https://orcid.org/0000-0001-9778-0354}{}

\author{Joseph Dorfer}{Graz University of Technology}{joseph.dorfer@tugraz.at}{https://orcid.org/0009-0004-9276-7870}{Austrian Science Fund (FWF) 10.55776/DOC183.}

\author{Linda Kleist}{Universität Hamburg}{linda.kleist@uni-hamburg.de}{https://orcid.org/0000-0002-3786-916X}{}

\author{Torsten Ueckerdt}{Karlsruhe Institute of Technology}{torsten.ueckerdt@kit.edu}{https://orcid.org/0000-0002-0645-9715}{}

\author{Birgit Vogtenhuber}{Graz University of Technology}{birgit.vogtenhuber@tugraz.at}{https://orcid.org/0000-0002-7166-4467}{Austrian Science Fund (FWF) 10.55776/DOC183.}

\Copyright{H\aa vard Bakke Bjerkevik, Joseph Dorfer, Linda Kleist,\\ Torsten Ueckerdt, and Birgit Vogtenhuber}



\nolinenumbers

\hideLIPIcs

\begin{document}
\maketitle

\begin{abstract}
    We consider the problem of reconfiguring non-crossing spanning trees on point sets.
    For a set~$P$ of $n$ points in general position in the plane, the ﬂip graph $\mathcal{F}(P)$ has a vertex for each non-crossing spanning tree on $P$ and an edge between any two spanning trees that can be transformed into each other by the exchange of a single edge (coined a flip).
    This flip graph has been intensively studied, lately with an emphasis on determining its diameter $\diam(\mathcal{F}(P))$ for sets $P$ of $n$ points in convex position.
    For this case, the current best bounds are $\nicefrac{14}{9}\cdot n - O(1) \le \diam(\mathcal{F}(P)) < \nicefrac{15}{9}\cdot n - 3$, obtained in a recent breakthrough work [Bjerkevik, Kleist, Ueckerdt, and Vogtenhuber; SODA 2025].
    The crucial tool for both the upper and lower bound are so-called \emph{conflict graphs}, which the authors stated might be the key ingredient for determining the diameter (up to lower-order terms).

    In this paper, we pick up the concept of conflict graphs from the above-mentioned work and show that this tool is even more versatile than previously hoped. 
    As our first main result, we use conflict graphs to show that computing the flip distance between two non-crossing spanning trees is \NP-hard, even for point sets in convex position.
    Interestingly, the result still holds for more constrained flip operations, concretely, compatible flips (where the removed and the added edge do not cross) and rotations (where the removed and the added edge share an endpoint).
    
    Additionally, we present new insights on the diameter of the flip graph, by this directly extending the line of research from [BKUV SODA25]. 
    Their lower bound is based on a constant-size pair of trees, one of which is of a type we refer to as \emph{stacked}.
    We show that if one of the trees is stacked, then the lower bound is indeed optimal up to a constant term, that is, there exists a flip sequence of length at most $\nicefrac{14}{9}\cdot (n-1)$ to any other tree.

    Lastly, we improve the lower bound on the diameter of the flip graph $\mathcal{F}(P)$ for $n$ points in convex position to $\nicefrac{11}{7}\cdot n-o(n)$.
\end{abstract}

\newpage

%%%%%%%%%%%%%%%%%%%%%
%%                 %%
%%  INTRODUCTION   %%
%%                 %%
%%%%%%%%%%%%%%%%%%%%%
\section{Introduction}\label{sec:intro}

Non-crossing configurations such as triangulations, spanning trees, matchings, or polygonizations on planar point sets are fundamental structures in computational geometry, with wide applications. 
In dynamic environments or interactive settings, it is often necessary to transform one configuration into another through a sequence of local modifications, while ensuring that the non-crossing property is preserved at every step and that the intermediate configurations are of the same type.
This process, known as (discrete) reconfiguration, raises natural algorithmic questions:
Given two non-crossing configurations, is it possible to convert one into the other using only basic operations? If so, it is of interest to bound the length of a shortest reconfiguration sequence between any two configurations, to compute the length of an optimal sequence between two given configurations, and to efficiently determine a short(est) sequence.

These questions
can be nicely rephrased in terms of a so-called \emph{flip graph}.
The flip graph of a discrete reconfiguration problem has a vertex for each configuration and an edge between any two configurations that can be transformed into one another by a single basic discrete operation, called a \emph{flip}.
In terms of the flip graph, the above-mentioned three questions then read as follows: 
Is the flip graph connected?
What is the diameter of the flip graph?
Can a shortest path (flip sequence) between two vertices of the flip graph be computed efficiently?

For many types of geometric graphs, the most classical flip operation is the exchange of one edge with another edge, see also \cref{fig:intro}.
For instance, it is known that the flip graph of triangulations on a set of $n$ points in the plane on convex position~\cite{Lawson72} and has diameter exactly $2n-10$~\cite{sleator1986rotation,pournin2014diameter}.
For decades, it has been a tantalizing and important open problem in theoretical computer science whether a shortest flip sequence between two triangulations
of a convex point set
can be computed in polynomial time.
This problem is equivalent to the complexity of computing the rotation distance between two rooted binary trees.

In this work, we consider flips between non-crossing geometric spanning trees.
Let $P$ denote a set of $n$ points in the plane in general position, that is, no three points of $P$ are collinear. 
A \emph{non-crossing spanning tree on $P$} is a spanning tree that has $P$ as its vertex set and whose edges are pairwise non-crossing straight-line segments. 
A \emph{flip} in a non-crossing spanning tree $T$ on $P$ is the exchange of one edge of $T$ by a different one such that the resulting graph is again a non-crossing spanning tree $T'$ on $P$; see \cref{fig:intro} for an example.
The according ﬂip graph $\mathcal{F}(P)$ has a vertex for each non-crossing spanning tree on $P$ and an edge between any two trees whenever they can be transformed into each other by a single edge ﬂip.
For two trees $T$ and $T'$ on $P$, we denote their flip distance, the length of a shortest path  between $T$ and $T'$ in $\mathcal{F}(P)$, by $\dist(T,T')$.

\begin{figure}[b]
    \centering
    \includegraphics[page=13]{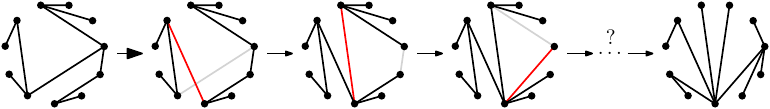}
    \caption{Reconfiguration of non-crossing spanning trees on convex point sets by edge exchange.
        The appearing edge is highlighted in red and the disappearing edge in gray.
        Can you complete the sequence to obtain the tree on the right?
        How many further flips are needed?
    }
    \label{fig:intro}
\end{figure}

In 1996, Avis and Fukuda~\cite{AvisFukuda} showed in their famous reverse search paper that $\mathcal{F}(P)$ is connected for any $n$-point set $P$ in general position, and that its radius is bounded from above by $n-2$, implying $\diam(\mathcal{F}(P)) \leq 2n-4$.
In 1999, Hernando, Hurtado, M{\'{a}}rquez, Mora, and Noy~\cite{Hernando} showed that for convex $n$-point sets $P$, it holds that $\diam(\mathcal{F}(P)) \geq \lfloor \nicefrac{3}{2}\cdot n\rfloor-5$. 
Since then, the flip graph $\mathcal{F}(P)$ of non-crossing geometric spanning trees on point sets $P$ in convex position has been subject of intensive study; see the discussion in \cref{sec:relatedwork} below for some details.
For a long time, it was conjectured that the true value for its diameter should be $\nicefrac{3}{2}\cdot n + c$ for some constant $c$.
Very recently, Bjerkevik, Kleist, Ueckerdt, and Vogtenhuber~\cite{bjerkevik2024flippingnoncrossingspanningtrees} refuted this conjecture by obtaining the currently best known bounds of $\nicefrac{14}{9}\cdot n+ O(1) \le \diam(\mathcal{F}(P)) \le \nicefrac{15}{9}\cdot n - O(1)$ for the diameter of the flip graph on non-crossing spanning trees on point sets in convex position.
The crucial tool for both the upper and lower bound are so-called \emph{conflict graphs}, which the authors introduced and stated might be the key ingredient for determining the diameter of the flip graph of spanning trees for point sets in convex position (up to lower-order terms).
Their lower bound is based on a constant-size example containing a tree belonging to a family of trees that we call \emph{stacked}. 

In this work, we pick up the concept of conflict graphs from~\cite{bjerkevik2024flippingnoncrossingspanningtrees} and show that this tool is even more versatile than previously hoped.
We prove that the flip distance between two trees can be as high as $\nicefrac{11}{7}\cdot n-o(n)$, by this directly extending the line of research from~\cite{bjerkevik2024flippingnoncrossingspanningtrees}.

\begin{theorem}
    For any set $P$ of $n \geq 3$ points in convex position, the flip graph $\mathcal{F}(P)$ of non-crossing spanning trees on $P$ has diameter at least $ (\nicefrac{11}{7}-o(1))n$.
\end{theorem}

We further show that if one of the trees is stacked, then the lower bound from~\cite{bjerkevik2024flippingnoncrossingspanningtrees} is indeed the best possible.
In other words, whenever one of the trees is stacked, then there exists a flip sequence of length at most $\nicefrac{14}{9}\cdot n- O(1)$ to any other tree.

\begin{theorem} \label{thm:stackedTrees}
    Let $T$, $T'$ be non-crossing trees on $n\geq 3$ points in convex position.
    If $T$ is a stacked tree, then $\text{dist}(T,T') \leq \frac{14}{9}(n-1)$.
\end{theorem}

Moreover, we use conflict graphs beyond the expectations of~\cite{bjerkevik2024flippingnoncrossingspanningtrees}, showing that they can also be used to prove that computing the flip distance between two non-crossing spanning trees is \NP-hard, even for point sets in convex position. 

\begin{restatable}{theorem}{hardness}\label{thm:hardness}
    Given two non-crossing spanning trees $T,T'$ on a point set  and an integer $\delta$, the decision problem of whether there exists a flip sequence of length at most~$\delta$ between $T$ and $T'$ is \NP-complete,
   even for point sets in convex position.
\end{restatable}

Interestingly, this result still holds for more constrained flip operations, concretely, \emph{compatible} flips (where the removed and the added edge do not cross) and \emph{rotations} (where the removed and the added edge share an endpoint). 
In \cref{fig:intro}, while the first flip is just an edge exchange, the second flip a compatible flip, and the third flip is a rotation.

\begin{restatable}{theorem}{hardness_restricted}\label{thm:hardness_restricted_flips}
        Given two non-crossing spanning trees $T,T'$ on a point set and an integer $\delta$, the decision problem of whether there exists a compatible flip sequence / a rotation sequence of length at most~$\delta$ between $T$ and $T'$ is \NP-complete,
        even for point sets in convex position.
\end{restatable}

We point out that the latter two statements constitute the first hardness results for flip graphs on point sets in convex position.
Moreover, \cref{thm:hardness_restricted_flips} refutes the long-held belief that the flip distance can be computed in polynomial time when the so-called happy edge property holds, because that property holds for compatible flips of non-crossing spanning trees on convex point sets.
See \cref{sec:relatedwork} for more background and \cref{table:2} for an overview of complexity status for various flip graphs of non-crossing structures on point sets.
\begin{table}[htb]
    \caption{Overview of complexity results. Results obtained in this work are marked by [*].}
    \begin{tabular}{|| c c || c c ||} 
        \hline
        setting & flip type & convex position & general position\\
        \hline\hline
        triangulations & edge exchange & 
        %open\footnotemark 
        \NP-hard?~\cite{dorfer2026flip}\footnotemark
        & \NP-hard ~\cite{LubiwP15,pilz2014flip}\\
        \multirow{4}{*}{non-crossing spanning trees} & edge exchange & \NP-hard [*] & \NP-hard [*]  \\
        & compatible & \NP-hard [*] & \NP-hard [*] \\
        & rotation & \NP-hard [*] & \NP-hard [*] \\
        & slide & open & open \\
        non-crossing spanning paths & edge exchange & in \P\ ~\cite{aichholzer2025linear}\, & open \\
        non-crossing perfect matchings & edge exchange & in \P~\cite{matchings} & \NP-hard \cite{MatchingHardness}\\
        non-crossing odd matchings & edge exchange & open & \NP-hard\ ~\cite{aichholzer2025flippingoddmatchingsgeometric}\\
        \hline
    \end{tabular}
    \label{table:2}
\end{table}

\pagebreak
%%%%%%%%%%%%%%%%%%%%%
%%     OUTLINE     %%
%%%%%%%%%%%%%%%%%%%%%
\subparagraph{Outline.}
We discuss further
related work in \Cref{sec:relatedwork} and introduce the central concepts from~\cite{bjerkevik2024flippingnoncrossingspanningtrees} that are required for our work in \Cref{sec:concepts}. 
In \Cref{sec:np-completeness}, we prove \NP-completeness for deciding the length of a shortest (general, compatible, or rotational) flip sequence between two spanning trees on a convex point set. 
In \Cref{sec:stacked}, we show that the flip distance between two trees on a convex $n$-point set is at most $\nicefrac{14}{9}\cdot(n-1)$ if one of them is stacked.
In \Cref{sec:newlowerbound}, we improve the lower bound on the diameter of the tree flip graph for $n$ points in convex position to $(\nicefrac{11}{7}-o(1))n$.
We conclude with a summary and open problems in \cref{sec:conclusion}.

%%%%%%%%%%%%%%%%%%%%%
%%  RELATED WORK   %%
%%%%%%%%%%%%%%%%%%%%%
\subsection{Related work}\label{sec:relatedwork}
We review the related work in two aspects, concerning the connectedness and diameter bounds on the one hand, and the computational complexity of computing a shortest flip sequence on the other hand.

\footnotetext{In a very recent arXiv preprint, Dorfer~\cite{dorfer2026flip} announced that computing the flip distance between triangulations of convex point sets is NP-complete. 
This (potential since not yet reviewed) result was inspired by and initiated after the submission of the work at hand, and uses similar techniques: 
The author introduces a notion of conflict graphs for triangulations of convex point sets and shows that the flip distance of certain pairs of triangulations is closely related to the size of largest acyclic subsets in the according conflict graph. 
\NP-hardness of the flip distance problem is then shown via proving \NP-hardness of finding largest acyclic subsets in these conflict graphs.}

%%%%%%%%%%%%%%%%%%%%%%%%%%%%%
%% CONNECTEDNESS, DIAMETER %%
%%%%%%%%%%%%%%%%%%%%%%%%%%%%%
\subparagraph{Connectedness and Diameter.}
For triangulations, connectedness of the flip graph for point sets was shown by Lawson~\cite{Lawson72} and for simple polygons by Hurtado, Noy, and Urrutia~\cite{HurtadoNU99}.
In both cases, the diameter of the flip graph is in $O(n^2)$. 
Hurtado, Noy, and Urrutia~\cite{HurtadoNU99} provided a matching lower bound of $\Omega(n^2)$ for both settings. 
For triangulations of convex point sets, Sleater, Tarjan, and Thurston~\cite{sleator1986rotation} showed that the diameter of the flip graph is at most $2n-10$, which is tight for $n>13$ as proven by Pournin~\cite{pournin2014diameter}.
There are many more results on flips in triangulations; see for example \cite{Eppstein10,Eppstein10_socg07,rainbowJournal,rainbowSoCG18,HouleHNR05,WagnerW22,WagnerW22_socg20,WagnerW22_soda20}. 

For non-crossing spanning trees of point sets in general position, Avis and Fukuda~\cite{AvisFukuda} showed that the flip graph is connected for flips, compatible flips, and rotations, and has a diameter at most $2n-4$. 
Nichols, Pilz, T\'oth, and Zehmakan~\cite{TreeTransition} gave an upper bound on the diameter of $O(n\log(n))$ for empty triangle rotations. 
Aicholzer, Aurenhammer, and Hurtado~\cite{AichholzerAH02} proved that the flip graph is connected for edge slides (where the source and the target edge together with some other edge of the tree form an empty triangle).
They also provided an exponential upper bound on the diameter, which was later improved to $O(n^2)$ by Aichholzer and Reinhart~\cite{aichholzer2007quadratic}. 
When considering point sets in convex position, determining the exact diameter of the flip graph has gained considerable attention in the last couple of years.
Starting from the benchmark upper bound of $2n-4$~\cite{AvisFukuda} and the conjectured to be essentially tight lower bound of $\lfloor\nicefrac{3n}{2}\rfloor - 5$~\cite{Hernando}, the upper bound was subsequently improved to $2n - \log(n)$ \cite{aichholzer2022reconfiguration} and $2n-\sqrt{n}$ \cite{bousquet2023noteJOURNAL}.
The leading constant factor of~$2$ was broken by Bousquet, de Meyer, Pierron, and Wesolek~\cite{bousquet2024reconfigurationSoCG}, who improved the upper bound to $1.96n$.
The latest advancement before the work at hand is due to the recent work~\cite{bjerkevik2024flippingnoncrossingspanningtrees}, which improved the upper bound to $\nicefrac{5}{3}\cdot n-3$ and the lower bound to $\nicefrac{14}{9}\cdot n-O(1)$.
We note that the lower bound of $\nicefrac{14}{9}\cdot n-O(1)$ is not only the best bound for points in general position, but also holds for restricted flip types.
When restricting the flips to compatible flips or rotations in the convex setting, the best standing upper bounds on the diameter of the flip graphs of spanning trees are $\nicefrac{5}{3}\cdot n-2$ and $\nicefrac{7}{4}\cdot(n-1)$, respectively, due to Aichholzer, Dorfer, and Vogtenhuber~\cite{aichholzer2025constrainedflipsplanespanning}. 

For plane spanning paths, connectedness of the flip graph is an interesting open problem.
So far, the flip graph has been shown to be connected for convex point sets~\cite{akl2007planar} (with diameter $2n-6$ for $n>4$ \cite{convexDiameter}), wheel sets and double circles~\cite{2023Aicholzer}, as well as point sets with two convex layers~\cite{KKR_paths}.
For matchings, the connectedness problem becomes a parity issue: While the flip graph of odd matchings is known to be connected~\cite{aichholzer2025flippingoddmatchingsgeometric}, the connectedness of the flip graph remains an open question for perfect matchings.
Partial progress has been made for convex point sets by Hernando, Hurtado, Noy~\cite{matchings} and for the case when an unbounded number of edges can be exchanged by Houle, Hurtado, Noy, and Rivera-Campo~\cite{HouleHNR05}.

%%%%%%%%%%%%%%%%%%%%%%%%%%%%%
%% COMPLEXITY, HAPPY EDGES %%
%%%%%%%%%%%%%%%%%%%%%%%%%%%%%
\subparagraph{Computational complexity of shortest flip sequences and the happy edge property.}
In several settings, there is a close connection between the existence of \NP-hardness proofs or polynomial-time algorithms and the happy edge property.
A reconfiguration problem has the \emph{happy edge property} if edges that are in the initial as well as the target configuration (and thus are \emph{happy}) 
can remain in the configuration throughout the entire flip sequence.

The happy edge property holds for triangulations of point sets in convex position~\cite{sleator1986rotation}, but not for triangulations of general point sets or polygons~\cite{bose2009flipsinplanar}.
The latter fact was used to show that the flip distance problem between triangulations of point sets and polygons is \NP-complete~\cite{LubiwP15,triangulationPolygon} and also \APX-hard on point sets~\cite{pilz2014flip}.
In contrast, on point sets without empty convex pentagons, the flip distance can be computed in polynomial time~\cite{eppstein2007happy}. 
If we consider triangulations of point sets in convex position, 
the flip graph is the 1-skeleton of a famous polytope, the \emph{associahedron}.

For the more general case of graph associahedra, where incidence relations are given by small reconfiguration steps between tubings of graphs~\cite{devadoss2009realization}, the shortest path problem has recently been shown to be \NP-hard~\cite{graphAssoHardness}.
The underlying graphs of the graph associahedra that appear in the hardness proof have high degree vertices. 
The graph associahedron whose underlying graph is a path is the associahedron,
for which the flip distance problem might also be \NP-hard
by the very recent arXiv preprint~\cite{dorfer2026flip}.

In the case of perfect matchings on general point sets, \NP-hardness of computing a shortest flip sequence was shown via counterexamples to the happy edge property~\cite{MatchingHardness}.
For convex point sets, the happy edge property holds and the flip distance between two matchings can be determined in linear time~\cite{matchings}.
Similarly, for the reconfiguration of odd matchings on general point sets, the happy edge property also does not hold and the flip distance is \NP-hard to compute~\cite{aichholzer2025flippingoddmatchingsgeometric}.

All these results support the long standing belief that the happy edge property is tied to the complexity of graph reconfiguration problems: 
If the happy edge property holds, then the problem is supposedly easy to solve.
If it doesn't hold, the problem should be hard to solve.

Most recently, this pattern has been broken, when Aichholzer and Dorfer~\cite{aichholzer2025linear} provided a linear time algorithm to determine the flip distance between non-crossing spanning paths of convex point sets, despite the existence of counterexamples to the happy edge property. 
What remained open (until now) is the other direction: 
Is there a graph reconfiguration problem for which it is \NP-hard to compute the flip distance, even though the problem has the happy edge property? 
We give an affirmative answer to this question in this work: 
The happy edge property holds for compatible flips on non-crossing spanning trees on convex point sets~\cite{aichholzer2025constrainedflipsplanespanning}, but the flip distance is \NP-hard to compute as shown in \cref{{thm:hardness_restricted_flips}}. 

%%%%%%%%%%%%%%%%%%%%%
%%                 %%
%%  FUNDAMENTALs   %%
%%                 %%
%%%%%%%%%%%%%%%%%%%%%
\section{Fundamental concepts: Edge pairs, conflict graphs, and blowups}
\label{sec:concepts}
			
Our results build on insights from Bjerkevik, Kleist, Ueckerdt, and Vogtenhuber~\cite{bjerkevik2024flippingnoncrossingspanningtrees}.
In the following, we introduce some of their concepts. 

%%%%%%%%%%%%%%%%%%%%%%%%%%%
%% LINEAR REPRESENTATION %%
%%%%%%%%%%%%%%%%%%%%%%%%%%%
\subparagraph*{Linear Representation.}
A convex point set is usually visualized by equally spaced points on a circle. 
A \emph{linear representation} of a non-crossing tree $T$ on a convex point set is constructed by cutting open the circle and unfolding the points into a horizontal line segment, called the spine.
Edges in $T$ can be illustrated nicely by semi-circles connecting their end points above the spine.
Given two trees (for example the initial tree and the target tree of a flip sequence) in the same linear representation, it is convenient to represent the edges of one tree above the spine and the edges of the other tree below the spine.
For an example, see \cref{fig:cut-open}.
If an edge between two consecutive vertices on the spine is contained in both trees we will in later parts of the paper depict this by drawing it as a straight line between the two vertices instead of two semicircles.

\begin{figure}[ht]
    \centering
    \begin{subfigure}{.3\textwidth}
    	\centering
    	\includegraphics[page=1]{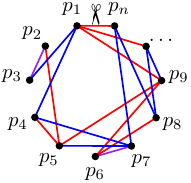}
    	\subcaption{}
    	\label{fig:cut-openA}
    \end{subfigure}\hfil
    \begin{subfigure}{.4\textwidth}
    	\centering
    	\includegraphics[page=2]{figures/Intro-2}
    	\subcaption{}
    	\label{fig:cut-openB}
    \end{subfigure}
    \caption{
        Figure reproduced from \cite{bjerkevik2024flippingnoncrossingspanningtrees}.
        \subref{fig:cut-openA} Two non-crossing trees $T,T'$ on a circularly labeled point set in convex position and \subref{fig:cut-openB} its linear representation with $T$ above and $T'$ below the spine.
    }
    \label{fig:cut-open}
\end{figure}

By labeling the vertices along the linear representation by $p_1,...,p_n$, we get a natural notion of the \emph{length} of an edge. 
Concretely, the length of the edge $p_ip_j$ is defined as $\lvert i-j\rvert$. 
We say a point $p_k$ is \emph{covered} by an edge $p_ip_j$ with $i<j$ if $i\leq k\leq j$. 
An edge \emph{covers} another edge if it covers both of its endpoints.

%%%%%%%%%%%%%%%%%%%
%% GAPS, PAIRING %%
%%%%%%%%%%%%%%%%%%%
\subparagraph*{Gaps and Edge Pairing.}
A \emph{gap} in the convex hull is the open line segment between two consecutive points $p_i$ and $p_{i+1}$ on the spine.
For every tree $T$, there is a natural \emph{gap-edge bijection} between gaps in the linear representation and edges of $T$:
Each gap is assigned to the shortest edge that covers the gap.
This assignment is indeed a bijection by Lemma 3.1 in~\cite{bjerkevik2024flippingnoncrossingspanningtrees}.
For an example, consider \cref{fig:IntroPairsA}.
Based on the assigned gap, edges of $T$ are partitioned into three groups. 
In particular, an edge is \emph{short} if it shares both vertices with its gap, \emph{near} if it shares one vertex with its gap, and \emph{long} if shares no  vertex with its gap.

\begin{figure}[ht]
    \centering
    \begin{subfigure}
    {.4\textwidth}
     \centering
     \includegraphics[page=3]{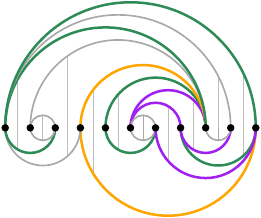}
    	\subcaption{}
    	\label{fig:IntroPairsA}
    \end{subfigure}\hfil
    \begin{subfigure}{.18\textwidth}
    	 \centering
    \includegraphics[page=6]{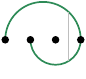}
    	\subcaption{}
    	\label{fig:IntroPairsB}
    \end{subfigure}\hfil
    \begin{subfigure}{.18\textwidth}
    	 \centering
    \includegraphics[page=7]{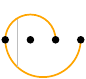}
    	\subcaption{}
    	\label{fig:IntroPairsC}
    \end{subfigure}\hfil
    \begin{subfigure}{.18\textwidth}
    	 \centering
    \includegraphics[page=8]{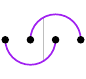}
    	\subcaption{}
    	\label{fig:IntroPairsD}
    \end{subfigure}\hfil
    \caption{
        \subref{fig:IntroPairsA} Illustration of the gap-edge and edge-edge bijection.
        %Near-near pairs are colored.
        Examples of an above \subref{fig:IntroPairsB}, a below \subref{fig:IntroPairsC}, and a crossing pair \subref{fig:IntroPairsD}.
        Edges belonging to near-near pairs in $A$, $B$, and $C$ are colored green, orange, and purple, respectively (also in several other figures).
    }
    \label{fig:IntroPairs}
\end{figure}

For two trees $T$ and $T'$, we obtain an edge-edge bijection by pairing two edges that are assigned to the same gap of the linear representation.
The authors of~\cite{bjerkevik2024flippingnoncrossingspanningtrees} explain how to handle nine types of edge pairs.
For our results it is sufficient to concentrate on near-near pairs where the two edges are distinct.
We call a gap $g_i$ with edges $e_i$ from $T$ and $e'_i$ from $T'$
\begin{description}
    \item[above]  if $e_i$ and $e'_i$ are adjacent and $e_i$ is longer than $e'_i$, see \cref{fig:IntroPairsB}.
    \item[below] if $e_i$ and $e'_i$ are adjacent and $e_i$ is shorter than $e'_i$, see \cref{fig:IntroPairsC}.
    \item[crossing] if $e_i$ and $e'_i$ are not adjacent, equivalently, the two edges cross if drawn in the same convex point set, see \cref{fig:IntroPairsD}.
\end{description}
The set of all above, below, and crossing gaps is denoted by $A$, $B$, and $C$, respectively.

%%%%%%%%%%%%%%%%%%%%%
%% CONFLICT GRAPHS %%
%%%%%%%%%%%%%%%%%%%%%
\subparagraph*{Conflict Graphs.}
\emph{Conflict graphs} were introduced as a tool to find the largest set of near-near pairs where the initial edge can be flipped directly to the target edge.
The conflict graph $H$ has a vertex for each (gap associated to a) near-near pair and a directed edge from a pair $(e_i,e'_i)$ to $(e_j,e'_j)$ if the edge $e_i$ needs to be removed before the flip that adds $e_j'$ and removes $e_j$ can be performed.
We denote the set of vertices in $H$ by $V(H)$.
There exist three types of conflicts. Specifically, there is a directed edge $\overrightarrow{g_ig_j}$ in $H$, from $g_i$ to $g_j$ if
\begin{description}
    \item[type 1] $e_i$ crosses $e'_j$, see \cref{fig:conflict-edgeA}.
    \item[type 2] $e'_j$ covers $e_i$ and $e_i$ covers $g_j$, see \cref{fig:conflict-edgeB}.
    \item[type 3] $e_i$ covers $e'_j$ and $e'_j$ covers $g_i$, see \cref{fig:conflict-edgeC}.
\end{description}

\begin{figure}[ht]
    \centering
    \begin{subfigure}{.24\textwidth}
    	\centering
    	\includegraphics[page=18]{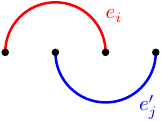}
        \subcaption{}
    	\label{fig:conflict-edgeA}
    \end{subfigure}\hfil
    \begin{subfigure}{.24\textwidth}
    	\centering
    	\includegraphics[page=19]{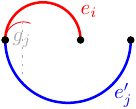}
        \subcaption{}
    	\label{fig:conflict-edgeB}
    \end{subfigure}\hfil
    \begin{subfigure}{.24\textwidth}
    	\centering
    	\includegraphics[page=20]{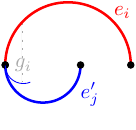}
        \subcaption{}
    	\label{fig:conflict-edgeC}
    \end{subfigure}
    \caption{
        Figure reproduced from \cite{bjerkevik2024flippingnoncrossingspanningtrees}.
        Examples of conflicts: \subref{fig:conflict-edgeA} type~1, \subref{fig:conflict-edgeB} type~2,  \subref{fig:conflict-edgeC} type~3.
    } 
    \label{fig:conflict-edge}
\end{figure}

We slightly abuse terminology and view the vertices of $H$ interchangeably as gaps and as near-near pairs.
Let $\las(H)$ denote the largest number of vertices in $H$ which induce an acyclic subgraph.
The authors of~\cite{bjerkevik2024flippingnoncrossingspanningtrees} derive the following relations between acyclic sets in the conflict graph and flip distances.
        
\begin{theorem}[{\!\cite[Theorem 2.1]{bjerkevik2024flippingnoncrossingspanningtrees}}]
    \label{thm:bigtheorem}
    Let $T$, $T'$ be two non-crossing trees on linearly ordered points $p_1,\ldots,p_n$ with corresponding conflict graph $H=H(T,T')$.
    \begin{enumerate}[(i)]
        \item If $V(H)$ is non-empty, then $\dist(T,T')\leq \max\{\frac{3}{2}, 2-\frac{\las(H)}{\lvert V(H) \rvert}\}(n-1)\}$. If $V(H)$ is empty, then $\dist(T,T')\leq \frac{3}{2}(n-1)$.\label{item:big-UB}
        \item If $V(H)$ is non-empty, then there is a constant $c$ depending only on $T$ and $T'$ such that for all $\bar{n} \geq 1$, we have $\diam(\mathcal{F}_{\bar{n}}) \geq \left(2-\frac{\las(H)}{\lvert V(H) \rvert}\right)\bar{n}-c$.\label{item:big-LB}
    \end{enumerate}
\end{theorem}
	
Moreover, we will use the following lemma on the sets of above, below, and crossing gaps.

\begin{lemma}[{\!\cite[Lemma~3.2]{bjerkevik2024flippingnoncrossingspanningtrees}}]
    \label{ac:ABC_SODA}
    Each of $A$, $B$, and $C$ is an acyclic set of $H$.
\end{lemma}
    
%%%%%%%%%%%%%%%%%%%%
%%    BLOWUPS    %%
%%%%%%%%%%%%%%%%%%%%
\subparagraph*{Blowups.}
By default, a flip sequence based on the described edge-edge bijection is not necessarily a shortest one and can in fact be very far from optimal.
The authors of~\cite{bjerkevik2024flippingnoncrossingspanningtrees} overcome this by introducing \emph{blowups} for near-near pairs. 
For trees $T$ and $T'$, the $\beta$-blowup is a pair of trees $\beta\cdot T$ and $\beta\cdot T'$ obtained by the following construction. 
For every gap $g$ that corresponds to a near-near pair $(e,e')$, insert a set $V(e)=V(e')$ of $\beta$ additional points in the gap~$g$.
In $\beta\cdot T$ add an edge from each $v\in V(e)$ to the endpoint of $e$ that is not adjacent to $g$. 
Proceed similarly in $\beta\cdot T'$ for $e'$.
Part \eqref{item:big-LB} of \cref{thm:bigtheorem} is obtained by considering blowups.
    
\begin{lemma}[{\!\cite[Lemma 23, arXiv-Version]{bjerkevik2024flippingnoncrossingspanningtrees}}]
    \label{lem:lower}
    The flip distance from $\beta\cdot T$ to $\beta\cdot T'$ is at least~$(\beta-2n)(2\lvert V(H)\rvert -  \las(H) )$, where $H = H(T,T')$.
\end{lemma}

%%%%%%%%%%%%%%%%%%%%%
%%                 %%
%% NP-COMPLETENESS %%
%%                 %%
%%%%%%%%%%%%%%%%%%%%%
\section{NP-Completeness results}
\label{sec:np-completeness}
In this section, we show that the decision problem of whether there exists a flip sequence of length at most $\delta$ between two given non-crossing spanning trees $T$ and $T'$ is \NP-complete, even for point sets in convex position.
At the end of this section, we prove that computing the flip distance remains \NP-complete when restricting the flip operation to compatible flips or rotations, respectively.

\hardness*

Our proof of \cref{thm:hardness} goes via an intermediate problem.
Specifically, we show that it is \NP-complete to decide whether a directed graph $H$ has an acyclic subset of size at least~$k$, even if $H$ is the conflict graph of two non-crossing trees with a given linear representation.

\begin{restatable}{proposition}{lashardness}\label{prop:las-hardness}
    Given two non-crossing trees $T,T'$ on a convex point set, with a linear representation and corresponding non-empty conflict graph $H = H(T,T')$, as well as an integer $k$,
    it is \NP-complete to decide whether $\las(H) \geq k$.
\end{restatable}

Our proof of \cref{prop:las-hardness} is a reduction from a variant of the \NP-complete problem \textsc{Max-2SAT}, which is defined as follows.
We are given an integer $t$ and a 2\emph{\textsc{SAT}-formula}~$\phi$ on a set $X$ of variables and a set $\mathcal{C}$ of clauses.
That is, every clause $C \in \mathcal{C}$ consists of a conjunction of two literals, each being a negated or non-negated variable.
\textsc{Max-2SAT} now asks whether there exists a truth assignment of the variables in $X$ such that at least $t$ of the clauses in~$\mathcal{C}$ are satisfied.
Every instance of \textsc{Max-2SAT} yields a bipartite variable-clause \emph{incidence graph}~$G_\phi$ with vertex set $V(G_\phi)=X \cup \mathcal{C}$ and edge set $E(G_\phi)$ containing an edge between a variable $x \in X$ and a a clause $C \in \mathcal{C}$ if and only if $x$ appears (negated or non-negated) in $C$.
We use the NP-hard variant \textsc{Planar V-Cycle Max-2SAT}~\cite[Theorem~7]{buchin2020geometric}. This variant
is the restriction of \textsc{Max-2SAT} to formulas $\phi$ whose vertex-clause incidence graph $G_\phi$ admits an
so-called
\emph{aligned drawing}, namely, a drawing with all variables on a horizontal line $L$ and no edge crossing $L$ or any other edge (where the drawing is part of the input). 
In particular, every clause lies either in the upper or lower halfplane defined by $L$.
\cref{fig:reduction} depicts an example.

\begin{figure}[htb]
\centering
\includegraphics[page=5]{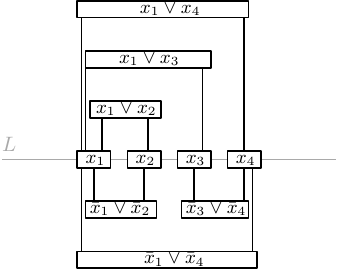}
\caption{
    An aligned drawing of a \textsc{Planar V-Cycle Max-2SAT} instance.
}
\label{fig:reduction}
\end{figure}

Let $t$ be a fixed integer and $\phi$ be a fixed \textsc{2SAT}-formula with variable set $X$, clause set~$\mathcal{C}$, planar incidence graph $G_\phi$, and an aligned drawing of~$G_\phi$.
Our task is to construct two non-crossing trees $T,T'$ with a linear representation, as well as an integer $k$, such that the following are equivalent:
\begin{itemize}
    \item The conflict graph $H = H(T,T')$ admits an acyclic subset of size at least $k$.
    \item The \textsc{2SAT}-formula $\phi$ admits a truth assignment that satisfies at least $t$ clauses.
\end{itemize}
We construct $T$ and $T'$ with a linear representation by modeling the variables and clauses in~$\phi$ with small gadgets.

In doing so, we also discuss the conflicts that occur between the gaps corresponding to the constructed edge pairs.
Recall that each gap has a type: above, below, or crossing.
In fact, it will be enough for us to focus on conflicts between gaps of different types, which we call \emph{mixed conflicts} for short.
Let us already remark that our construction only has conflicts $\overrightarrow{g_ig_j}$ of type~1, i.e., where $e_i$ properly crosses $e'_j$ as illustrated in \cref{fig:conflict-edgeA}.
In fact, for every below gap $g_j$ the common endpoint of its two edges will be a leaf in the tree $T$, and hence there is no conflict $\overrightarrow{g_ig_j}$ of type~2.
Similarly, for every above gap $g_i$ the common endpoint of its two edges will be a leaf in the tree $T'$, and hence there is no conflict $\overrightarrow{g_ig_j}$ of type~3.

%%%%%%%%%%%%%%%%%%%%%
%% VARIABLE-GADGET %%
%%%%%%%%%%%%%%%%%%%%%
\subparagraph{Variable-gadget.}
Let $x \in X$ be a variable and $d = \deg(x)$ be its degree in $G_\phi$, that is, $x$ occurs (negated or non-negated) in exactly $d$ clauses in $\phi$.
The variable-gadget for $x$ consists of $6d$ near-near pairs and $2d+1$ short-short pairs (i.e., $8d+1$ edges in each of the two trees) that we place at the position of $x$ on the horizontal line $L$ of the aligned drawing of $G_\phi$.
That is, the rightmost point of the variable-gadget for $x$ coincides with the leftmost point of the variable-gadget immediately to the right of it.

\begin{figure}[htb]%
    \begin{subfigure}[t]{\textwidth}%
        \centering%
        \includegraphics[page=1]{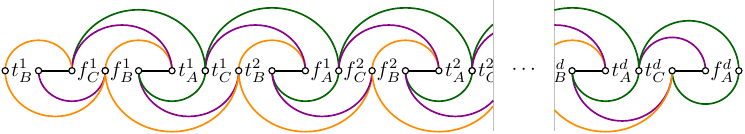}%
        \subcaption{}
            \label{fig:variable-gadget-trees}
    \end{subfigure}

    \begin{subfigure}[t]{\textwidth}%
        \centering%
        \includegraphics[page=3]{figures/variable-gadget.pdf}%
        \subcaption{}%
        \label{fig:variable-gadget-conflicts}
    \end{subfigure}%
    
    \caption{
        Illustration of a variable-gadget and the graph of its mixed conflicts.
    }
    \label{fig:variable-gadget}
\end{figure}

Recall that only the gaps corresponding to near-near pairs correspond to vertices in the conflict graph.
The $6d$ gaps corresponding to near-near pairs are organized in a sequence of $d$ groups of six gaps each.
For $i=1,\ldots,d$, the six gaps in the $i$-th group are denoted as $t^i_B,f^i_C,t^i_A,f^i_B,t^i_C,f^i_A$, where $t$ and $f$ stand for \emph{true} and \emph{false}.
The subscripts $A,B,C$ indicate the types \emph{above}, \emph{below}, and \emph{crossing} of the gaps (recall \cref{fig:cut-open} for an illustration).
For example, the gap $t^i_B$ is below, meaning that $t^i_B$ corresponds to a near-near pair $(e,e')$ of edges that are adjacent with $e' \in T'$ (in the lower halfplane) being longer than $e \in T$.

In the conflict graph, the gaps of a variable-gadget form a path of double conflicts, that is, any two consecutive of the $6d$ gaps in the sequence 
\[
    t^1_B,f^1_C,t^1_A,\;\; f^1_B,t^1_C,f^1_A, \quad t^2_B,f^2_C,t^2_A, \;\; f^2_B,t^2_C,f^2_A,\quad \ldots, \quad t^d_B,f^d_C,t^d_A, \;\; f^d_B,t^d_C,f^d_A
\]
have a double conflict.
See \cref{fig:variable-gadget-trees,fig:variable-gadget-conflicts} for an illustration.

In addition to these double conflicts, there are $6d-2$ further mixed conflicts as shown in \cref{fig:variable-gadget-conflicts}.
(Recall that a mixed conflict is a conflict between two gaps of different type.)
We observe that gaps of distinct variable-gadgets are not in conflict, while all mixed conflicts between two gaps in the same variable-gadget have the following crucial property.

\begin{observation}\label{obs:conflicts-variable-gadget}
    If $\overrightarrow{g_1 g_2}$ is a mixed conflict between a gap in a gadget for variable $x$ and a gap in a gadget for variable $y$, then $x = y$.
    Moreover, (at least) one of the following holds:
    \newline
    \textbullet \ $g_1$ is above \qquad
    \textbullet \ $g_2$ is below \qquad
    \textbullet \ $g_1$ and $g_2$ have a double conflict
\end{observation}

For each variable $x$, the conflict graph of the variable-gadget has two natural maximum acyclic sets, namely $t(x) = \{ t^i_A, t^i_B, t^i_C \mid i \in [d]\}$ and $f(x) = \{ f^i_A, f^i_B, f^i_C \mid i \in [d]\}$.
(We shall argue below that $t(x)$ and $f(x)$ are indeed acyclic.)
These two choices encode the two possible truth assignments of variable $x$, where $x = \textsc{True}$ corresponds to including $t(x)$ in an acyclic set and $x = \textsc{False}$ corresponds to including~$f(x)$.

Recall that the degree of variable $x$ in $G_\phi$ is $\deg(x) = d$.
Consider the above gaps in the variable-gadget for $x$ and all left endpoints of the corresponding edges in the upper halfplane.
Observe that there are exactly $2d$ such points and that each is accessible from above without crossing any edges of the variable-gadget (see \cref{fig:variable-gadget-trees}).
In fact, these are the left endpoints of the $2d$ crossing gaps $f^1_C,t^1_C,\ldots,f^d_C,t^d_C$.
Later, we shall use these $2d$ \emph{connection points} to connect the clause-gadgets for each clause $C$ that contains $x$ and lies in the upper halfplane of the aligned drawing of $G_\phi$.
We prepare $2d$ connection points in the gadget for variable $x$ as there are at most $d = \deg(x)$ clauses that contain $x$ negated and at most $d$ clauses that contain $x$ non-negated.
Symmetrically, the $2d$ right endpoints of the crossing gaps (the right endpoints of edges in the lower halfplane corresponding to the $2d$ below gaps) are accessible from below without crossing any edges, and we shall use these \emph{connection points} to connect clause-gadgets in the lower halfplane.

In the variable-gadget of each variable $x$ we choose one connection point $p_{x,C}$ for each incident clause $C$ such that $p_{x,C} \neq p_{x,C'}$ whenever $C \neq C'$, each $p_{x,C}$ is accessible from the side (upper or lower halfplane) that contains clause $C$, each $p_{x,C}$ is incident to a gap $f^i_C$ if $x$ appears negated in $C$ and incident to a gap $t^i_C$ if $x$ appears non-negated in $C$, and the variable-clause connections are non-crossing, i.e., respect the ordering around the vertices in the aligned drawing of~$G_\phi$.

%%%%%%%%%%%%%%%%%%%%%
%%  CLAUSE-GADGET  %%
%%%%%%%%%%%%%%%%%%%%%
\subparagraph{Clause-gadget.}

Let $C \in \mathcal{C}$ be a clause and let $x$ and $y$ be the two variables in $C$, with $x$ left of $y$ in the aligned drawing of $G_\phi$.
First, consider the case that clause $C$ lies in the upper halfplane.
Let $p_{x,C}$ and $p_{y,C}$ be the dedicated connection points in the variable-gadgets of $x$ and $y$.
We add new gaps $g_x$ and $g_y$ immediately right of $p_{x,C}$ and $p_{y,C}$, and two new edge pairs as illustrated in \cref{fig:aclause-above}.
Gap $g_x$ is crossing with an edge in $T$ to $p_{y,C}$ and an edge in $T'$ of length~$3$.
Gap $g_y$ is above with an edge in $T$ to $p_{x,C}$ and an edge in $T'$ of length~$3$.

\begin{figure}[htb]%
    \begin{subfigure}[t]{\textwidth}%
        \centering
        \includegraphics[page=1]{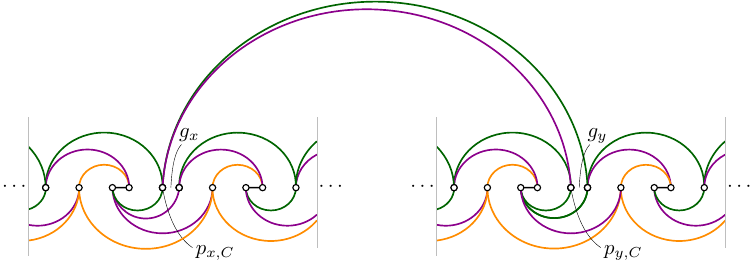}
        \subcaption{}%
        \label{fig:aclause-above}
    \end{subfigure}%

    \begin{subfigure}[t]{\textwidth}%
        \centering
        \includegraphics[page=2]{figures/clause-gadget.pdf}
        \subcaption{}%
        \label{fig:bclause-above}
    \end{subfigure}%
    
    \caption{
        Illustration of a clause-gadget for a clause above the horizontal line and its conflict graph.
        The special conflict corresponding to the last option in \cref{obs:conflicts-clause-gadget} is highlighted.
    }
    \label{fig:clause-gadget-above}
\end{figure}

In the (left) variable-gadget for $x$, we have that $g_x$ is in double conflict with the below gap immediately left of $p_{x,C}$.
This gap corresponds to $f^i_B$ (respectively $t^i_B$) for some $i$ if variable $x$ appears negated (respectively non-negated) in $C$.
Analogously, $g_y$ is in double conflict with the gap corresponding to $f^i_B$ or $t^i_B$ in the (right) variable-gadget for $y$, depending on whether variable $y$ appears negated or non-negated in $C$.
See \cref{fig:bclause-above} for an illustration of the mixed conflicts between the two gaps $g_x,g_y$ of the clause-gadget and the gaps of the variable-gadgets for $x$ and $y$.
Besides the mentioned double conflicts, there are six further mixed conflicts as shown in \cref{fig:bclause-above}.
We observe a crucial property of all these conflicts. 

\begin{observation}\label{obs:conflicts-clause-gadget}
    If $\overrightarrow{g_1 g_2}$ is a mixed conflict between a gap in the gadget for variable $x$ and a gap in the gadget for clause $C$, then $x$ appears in $C$.
    Moreover, one of the following holds:
    \newline
    \textbullet \ $g_1$ is above \qquad 
    \textbullet \ $g_2$ is below \qquad 
    \textbullet \ $g_1$ and $g_2$ have a double conflict
    \newline
    \textbullet \ there is a $g_1$-$g_2$-path of three double conflicts with three gaps of the variable-gadget for $x$
\end{observation}

Now consider the case that clause $C$ lies in the lower halfplane, in which case we proceed symmetrically.
Again, let $x,y$ be the variables incident to $C$ with $x$ left of $y$ and the designated connection points $p_{x,C}$ and $p_{y,C}$.
As illustrated in \cref{fig:aclause-below}, we introduce a new gap $g_x$ left of $p_{x,C}$ and a new gap $g_y$ left of $p_{y,C}$.
Gap $g_x$ is below with an edge to $p_{y,C}$ in $T'$ and an edge in $T$ of length~$3$.
Gap $g_y$ is crossing with an edge to $p_{x,C}$ in $T'$ and an edge in $T$ of length~$3$.

\begin{figure}[htb]%
    \begin{subfigure}[t]{\textwidth}%
        \centering
        \includegraphics[page=3]{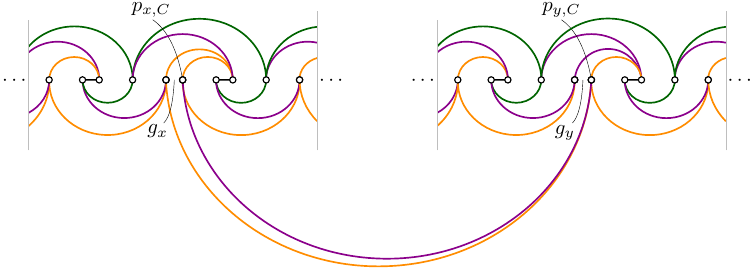}
        \subcaption{}%
        \label{fig:aclause-below}
    \end{subfigure}%

    \begin{subfigure}[t]{\textwidth}%
        \centering
        \includegraphics[page=4]{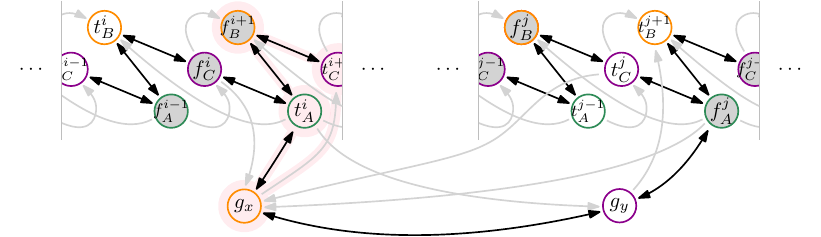}
        \subcaption{}%
        \label{fig:bclause-below}
    \end{subfigure}%
    
    \caption{
        Illustration of a clause-gadget for a clause below the horizontal line and its conflict graph.
        The special conflict corresponding to the last option in \cref{obs:conflicts-clause-gadget} is highlighted.
    }
    \label{fig:clause-gadget-below}
\end{figure}

As before, there are double conflicts between $g_x$ and $g_y$, between $g_x$ and the crossing gap immediately right of $p_{x,C}$, and between $g_y$ and the above gap immediately right of $p_{y,C}$.
Also, there are the six further mixed conflicts between $g_x,g_y$ and gaps in the variable-gadgets as shown in \cref{fig:bclause-below}.
Crucially, we observe that all these mixed conflicts also have the properties as described in \cref{obs:conflicts-clause-gadget}.

%%%%%%%%%%%%%%%%%%%%%%
%% PUTTING TOGETHER %%
%%%%%%%%%%%%%%%%%%%%%%
\subparagraph{Putting All Gadgets Together.}

This concludes the construction of all variable-gadgets and clause-gadgets.
In total, we obtain a linear representation with a set of edges $T_0$ in the upper halfplane and a set of edges $T_0'$ in the lower halfplane.
(The $0$-subscript indicates that these are not yet the final trees in our upcoming reduction.)
Since the underlying drawing of $G_\phi$ is planar, no two edges from clause-gadgets in the same halfplane are crossing.
Thus, each of $T_0$ and $T_0'$ is a non-crossing edge set.

Each edge pair $(e,e') \in T_0 \times T_0'$ is introduced with a corresponding gap $g$.
Observe that $e$ is the shortest edge in $T_0$ that covers gap $g$.
In particular, mapping each gap $g$ to the shortest edge in $T_0$ covering $g$ is a bijection.
This implies that $T_0$ is in fact a tree~\cite[Lemma~3.1]{bjerkevik2024flippingnoncrossingspanningtrees}.
By the same argument $T_0'$ is a non-crossing tree, i.e., $T_0,T_0'$ form a pair of non-crossing spanning trees with a linear representation.

Let $H_0 = H(T_0,T_0')$ be the corresponding conflict graph.
Its vertices are all gaps corresponding to near-near pairs of edges.
As already discussed, there are $6d$ such gaps for each variable $x$ of degree $\deg(x) = d$ in $G_\phi$, as well as two such gaps for each clause $C \in \mathcal{C}$.
\cref{obs:conflicts-variable-gadget,obs:conflicts-clause-gadget} already describe the mixed conflicts between two gaps in variable-gadgets and a gap in a variable-gadget with a gap in a clause-gadget.
The two gaps of the same clause-gadget form a double conflict.
It remains to consider a mixed conflict $\overrightarrow{g_1g_2}$ between a gap $g_1$ in the gadget for clause $C_1$ and a gap $g_2$ in the gadget for clause $C_2 \neq C_1$.
Recall that each of $g_1,g_2$ corresponds to an edge pair of which one edge has only length~$3$.
It follows that for $\overrightarrow{g_1g_2}$ to be a conflict (of type~3), their corresponding long edges must cross.
Hence, clause $C_1$ lies in the upper halfplane and $g_1$ is an above or crossing gap, while $C_2$ lies in the lower halfplane and $g_2$ is a below or crossing gap.

\begin{observation}\label{obs:conflicts-distinct-clause-gadget}
    If $\overrightarrow{g_1 g_2}$ is a mixed conflict between two gaps in clause-gadgets, then one of the following holds:
    \newline
    \textbullet \ $g_1$ is above \qquad 
    \textbullet \ $g_2$ is below \qquad 
    \textbullet \ $g_1$ and $g_2$ have a double conflict
\end{observation}

Summarizing \cref{obs:conflicts-variable-gadget,obs:conflicts-clause-gadget,obs:conflicts-distinct-clause-gadget}, we see that almost all mixed conflicts in $H_0 = H(T_0,T_0')$ are double conflicts, outgoing at an above gap or incoming at a below gap.
The only exceptions are the mixed conflicts between a clause-gadget and incident vertex-gadget as described in last item of \cref{obs:conflicts-clause-gadget}, which we call \emph{exceptional conflicts} for short.

\begin{lemma}\label{lem:reduction-acyclic-set}
    Let $S \subseteq V(H_0)$ be a subset of gaps such that there is no double conflict and no exceptional conflict between any two gaps in $S$.
    Then $S$ is an acyclic set in $H_0$.
\end{lemma}
\begin{proof}
    Let $A,B,C \subseteq V(H_0)$ denote the set of all above, below, and crossing gaps, respectively.
    It is known that each of $A,B,C$ is an acyclic set~\cite[Lemma~3.2]{bjerkevik2024flippingnoncrossingspanningtrees}.
    Fix topological orderings $t_A$ of $A \cap S$, $t_B$ of $B \cap S$, and $t_C$ of $C \cap S$.
    
    Let $\overrightarrow{g_1g_2}$ be any mixed conflict with $g_1,g_2 \in S$.
    By the assumptions on $S$ and \cref{obs:conflicts-variable-gadget,obs:conflicts-clause-gadget,obs:conflicts-distinct-clause-gadget} it holds $g_1 \in A$ or $g_2 \in B$ (or both).
    Hence, the concatenation $t_At_Ct_B$ is a topological ordering of $S$, i.e., $S$ is acyclic.    
\end{proof}

%%%%%%%%%%%%%%%%%%%%%
%%  THE REDUCTION  %%
%%%%%%%%%%%%%%%%%%%%%
\subparagraph{Reduction from {\textsc{Planar V-Cycle MAX-2SAT}}.} 

We are now ready for a full description of the desired reduction.
We are given a \textsc{2SAT}-formula $\phi$, an aligned drawing of the graph $G_\phi$, and an integer $t$.
First, we construct the pair $T_0,T'_0$ of non-crossing spanning trees with the linear representation consisting of variable-gadgets and clause-gadgets as described previously.
As a final modification, we apply blowups to the near-near gaps in each variable-gadget, while the two gaps of each clause-gadget remain unchanged.
Let $m = |\mathcal{C}|$ be the number of clauses in $\phi$.
In every variable-gadget we apply an $m$-blowup to its first and last gap, and a $2m$-blowup to all remaining $6d-2$ gaps corresponding to near-near pairs in the variable-gadgets.
See \cref{fig:app-reduction-blowup} for an illustration.

\begin{figure}[htb]
    \centering
    \includegraphics{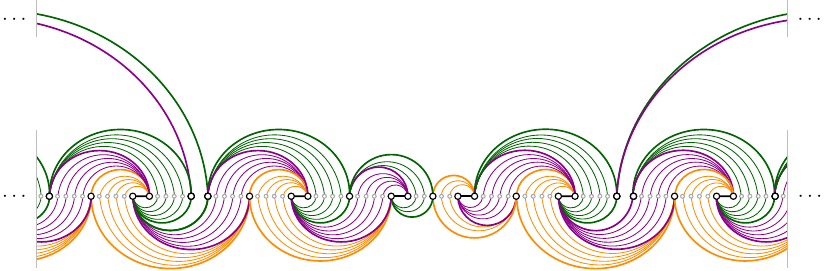}
    \caption{
        Applying an $m$-blowup (here $m = 2$) to the first and last gap of each variable-gadget and $2m$-blowups to all other gaps corresponding to near-near pairs in variable-gadgets.
    }
    \label{fig:app-reduction-blowup}
\end{figure}

It is easy to check that applying blowups maintains non-crossing trees, and we denote the resulting non-crossing pair of trees by $T,T'$.
Further let $H = H(T,T')$ be the corresponding conflict graph.
Every gap $g$ in the previous conflict graph $H_0 = H(T_0,T_0')$ corresponds to a set $B(g)$ of $1$, $m+1$, or $2m+1$ gaps in $H$, depending on the blowup applied to it.
Each gap in $B(g)$ has the same type (above, below, crossing) as $g$ and the same mixed conflicts as $g$.
(There are additional conflicts among the gaps in the same $B(g)$, but these are not mixed.)

In $H$ as well as $H_0$ there are exactly six gaps per clause, i.e., $6m$ such gaps in total.
For a variable $x \in X$ of degree $d = \deg(x)$ in $G_\phi$ there are exactly $6d$ gaps in $H_0$ corresponding to the variable-gadget for $x$.
Thus, the total number of gaps in variable-gadgets in $H_0$ is $\sum_{x\in X} 6\deg(x) = 6|E(G_\phi)| = 12|\mathcal{C}| = 12m$.
After the blowups, the number of gaps in variable-gadgets in $H$ is:
\[
     \sum_{x \in X}\left( (2m+1)6\deg(x)-2m\right) = (2m+1)6|E(G_\phi)|-2m|X| = 2m(12m+6-|X|)
\]
Setting $k = m(12m+6-|X|)+t$ completes our construction.
			
\begin{lemma}\label{lem:acyclic}
    The following are equivalent.
    \begin{itemize}
        \item The conflict graph $H = H(T,T')$ admits an acyclic subset of size at least $k$.
        \item The \textsc{2SAT}-formula $\phi$ admits a truth assignment that satisfies at least $t$ clauses.
    \end{itemize}
\end{lemma}
\begin{proof}
    First assume that $\phi$ admits an assignment that satisfies at least $t$ clauses.
    We shall first find a set $S_0 \subseteq V(H_0)$ of gaps, i.e., prior to the blowups.
    
    Consider any variable $x \in X$ and let $d = \deg(x)$ be its degree in $G_\phi$.
    Recall that the gaps in the variable-gadget for $x$ partition into $t(x) = \{ t^i_A, t^i_B, t^i_C \mid i \in [d]\}$ and $f(x) = \{ f^i_A, f^i_B, f^i_C \mid i \in [d]\}$ with no double conflict between two gaps in $t(x)$ or two gaps in $f(x)$.
    If $x$ is set to \textsc{True}, we add $t(x)$ to $S_0$, otherwise we add $f(x)$ to $S_0$.
    This way, we add exactly $\sum_{x \in X} 3\deg(x) = 3|E(G_{\phi})| = 6|\mathcal{C}| = 6m$ gaps of $H_0$ to $S_0$.

    Consider any clause $C \in \mathcal{C}$ and let $z$ be one of its incident variables. 
    Let $p_{z,C}$ be the designated connection point in the variable-gadget of $z$.
    Recall that $p_{z,C}$ is incident to a crossing gap in $t(z)$ if $z$ appears negated in $C$, and incident to a crossing gap in $f(z)$ if $z$ appears non-negated in $C$.
    Moreover, the corresponding gap $g_z$ in the clause-gadget for $C$ is in double conflict with a gap in $f(z)$ if $z$ appears negated in $C$, and in double conflict with a gap in $t(z)$ if $z$ appears non-negated in $C$.
    
    If $C$ is satisfied, we fix one variable $z$ that fulfills $C$.
    If $z$ appears non-negated in $C$, then $z$ is set to \textsc{True} and $t(z) \subseteq S_0$.
    In this case we add the corresponding gap $g_z$ of the clause-gadget to $S_0$.
    Note that $g_z$ is not in double conflict with any gap in $S_0$, and not in an exceptional conflict with any gap in $S_0$.
    Since at least $t$ clauses are satisfied, we add at least $t$ further gaps to $S_0$.
    Thus $|S_0| \geq 6m+t$, and $S_0$ is acyclic by \cref{lem:reduction-acyclic-set}.

    To obtain a set $S \subseteq V(H)$ of gaps after the blowups, we take for each gap $g \in S_0$ its corresponding set $B(g)$ after the blowups.
    As $S_0$ is acyclic in $H_0$, we have that $S$ is acyclic in $H$.
    For each variable $x$ the sets $t(x)$ and $f(x)$ contain exactly one gap that is the first or last in the variable-gadget, as well as $d-1$ further gaps.
    Consequently, $S$ contains exactly half of all the gaps in variable-gadgets in addition to the $t$ gaps in clause-gadgets.
    In other words, $|S| \geq m(12m+6-|X|) + t = k$, as desired.

    \medskip
    
    Now assume that $H = H(T,T')$ admits an acyclic subset $S \subseteq V(H)$ of size at least $k$.
    Consider the set $S_0 = \{ g \in V(H_0) \mid B(g) \cap S \neq \emptyset\}$ of all gaps in $H_0$ whose blowups contain an element of $S$.
    As $S$ is acyclic, $S_0$ is an acyclic subset of $H_0$.
    
    For a variable $x$ of degree $\deg(x) = d$, the double conflicts in the variable-gadget for $x$ form a path $P$ on $6d$ vertices.
    Thus, $S_0$ forms an independent set of $P$ and hence $S_0$ contains at most $3d$ gaps of the variable-gadget for $x$.
    As every independent set of $P$ contains at least one of $P$'s endpoints, it follows that $S$ contains at most half of the gaps in the variable-gadget of $x$.
    Thus $S$ contains at most half of the gaps in all variable-gadgets in $H$, i.e., at most $m(12m+6-|X|)$ such gaps.
    
    On the other hand, our assumption is that $|S| \geq k = m(12m+6-|X|) + t$.
    This means that $S$ contains at least $t$ gaps in clause-gadgets.
    Moreover, since $1 \leq t \leq |\mathcal{C}| = m$, it follows that the number of gaps in $S$ from variable-gadgets is at least $m(12m+6-|X|) - (m-1)$, i.e., at least $m-1$ less half of the total number of these gaps.
    However, if for some path $P$ of double conflicts in a variable-gadget we have $|S_0 \cap P| < d$ or $S_0$ contains both endpoints of $P$, then $S$ misses out on at least $m$ gaps, which is impossible.

    It follows that for each variable $x$ the set $S_0$ either contains all of $t(x)$ or all of $f(x)$.
    In the former case we set $x$ to \textsc{True}, in the latter case we set $x$ to \textsc{False}.
    Now consider any clause $C$ and any variable $z$ in $C$ with its corresponding gap $g_z$ in $H$.
    Recall that if $z$ appears negated (respectively non-negated) in $C$, then $g_z$ has a double conflict with a gap in $f(z)$ (respectively $t(z)$).
    Thus, if $S$ contains $g_z$, then the truth assignment of $z$ satisfies the clause $C$.
    Since $S$ contains at least $t$ such gaps (at most one of each clause-gadget due to their double conflicts), our truth assignment satisfies at least $t$ clauses, as desired.
\end{proof}
					
\cref{lem:acyclic} now immediately implies \cref{prop:las-hardness}.

\medskip

In order to finally prove \cref{thm:hardness}, we now reduce the problem of finding large acyclic subsets in conflict graphs to the problem of finding short flip sequences in $\beta$-blowups.

\begin{theorem}\label{thm:hardbeautiful}
    Let $T$ and $T'$ be two non-crossing spanning trees on $n$ vertices with a linear representation 
    {  and non-empty conflict graph $H = H(T,T')$.
    Further, let $k>0$ be an integer} 
    and let $\beta=4n^2+2n$. 
    Then the following are equivalent.
    \begin{itemize}
        \item The conflict graph $H$ has an acyclic subset of size at least $k$.
        \item The flip distance between $\beta\cdot T$ and $\beta\cdot T'$ is at most $(\beta+1)(2|V(H)| - k) + 2(n-|V(H)|)$.
    \end{itemize}
\end{theorem}
\begin{proof}
    For convenience, let us denote $|H| = |V(H)|$.
    Assume $H$ does not contain an acyclic set of size $k$, i.e., $\las(H) < k$.
    Then by Lemma~\ref{lem:lower},
    \begin{align*}
        \dist(\beta\cdot T, \beta\cdot T') & \geq (\beta-2n)(2|H| - \las(H)) \geq (\beta+1-(2n+1))(2|H| - k + 1)\\
        & = (\beta+1)(2|H| - k) + (\beta+1) - (2n+1)(2|H| - k + 1) \\
        &> (\beta+1)(2|H| - k) + 2(n-|H|).
    \end{align*}
    For the last inequality observe that $|H| < n$ and hence $(2n+1)(2|H|-k+1) \leq (2n+1)(2n-1) \leq 4n^2$,  which together with $\beta = 4n^2 + 2n > 4n^2 + 2(n-|H|)$ gives the desired lower bound.

    Conversely, assume $H = H(T,T')$ has an acyclic subset of size $k$, i.e., $\las(H) \geq k$.
    Then, the conflict graph $H^\beta = H(\beta\cdot T, \beta \cdot T')$ of the blowup has an acyclic set of size $k\cdot(\beta+1)$, i.e., $\las(H^\beta) \geq k(\beta+1)$.
    For the original trees $T,T'$ we have $n$ vertices and $n-1$ gaps, of which $|H|$ are affected by the blowup.
    Thus, each of $\beta\cdot T$ and $\beta \cdot T'$ has $n -1 + \beta |H|$ edges.
    Moreover, $|V(H^\beta)| = (\beta+1)|H|$.
    {
    Using one flip for each edge pair in the largest acyclic subset and two flips for each other edge pair gives the desired bound of 
    \begin{align*}
        \dist(\beta\cdot T, \beta\cdot T') \leq (\beta+1)(2|H|-\las(H)) + 2 (n-1-|H|) < (\beta+1)(2|H|-k) + 2 (n-|H|).
    \end{align*}
}
\end{proof}
			
We remark that the proof of \cref{thm:hardbeautiful} produces a flip sequence that contains non-compatible flips.
Next, we modify the proof to the case of flip sequences that contain only compatible flips or only rotations. 
In fact, a direct flip for an above pair or a below pair is already a rotation.
Further, flipping any near edge to a short edge covering its associated gap is also a rotation.
We do, however, need to deal with large collections of crossing edge pairs. To this end, we use the following fact.

\begin{lemma}\label{lemma:rotbound}
    Let $T$, $T'$ be two trees that differ in a single crossing pair. Then the rotation flip distance between $\beta\cdot T$ and $\beta\cdot T'$ is at most $\beta+2$. 
\end{lemma}
\begin{proof}
    We denote the vertices incident to the blowup edges by $u_0,\dots, u_\ell$ where $\ell=\beta+3$.
    Tree $\beta\cdot T$ contains the edges $u_iu_\ell$ for $i\in\{1,\dots, \ell -2\}$ and $\beta\cdot T'$ contains the edges $u_0u_i$ for $i\in\{2,\dots, \ell -1\}$.
    For an illustration, see \cref{fig:rotations}.
    In a first step, we temporarily flip the edge $u_1u_\ell$ to the edge $u_0u_\ell$.
    This is a rotation as $u_\ell$ is a common vertex.
    Then, for increasing $i\in\{2,\dots, \ell -2\}$, we flip $u_iu_\ell$ to $u_0u_i$.
    In a last step, we flip the edge $u_0u_\ell$ to $u_0u_{\ell-1}$.
    The total number of rotation flips is $\ell-1=\beta+2$.
    \begin{figure}[htb]%
        \captionsetup[subfigure]{justification=centering}%
        \begin{subfigure}[t]{\columnwidth}%
            \centering%            \includegraphics[page=17]{figures/Trees_Hardness2}%
        \end{subfigure}
        \caption{
            Illustration for the proof of \cref{lemma:rotbound}.
            The tree $\beta\cdot T$ can be transformed into $\beta\cdot T'$ by $\beta+2$ rotation flips.
            The current tree is illustrated above the spine and the edges that remain to be inserted (into the current top tree) are depicted below the spine.
        }
        \label{fig:rotations}
    \end{figure}
\end{proof}

With \cref{lemma:rotbound} in place, we prove that finding short compatible flip sequences and short rotation sequences is \NP-hard as well.
As before, we use the shorthand notation $|H| = |V(H)|$.

\begin{theorem}\label{thm:hardbeautifulcr}
    Let $T$ and $T'$ be two non-crossing spanning trees on $n$ vertices with a linear representation and non-empty conflict graph $H = H(T,T')$. Let $k>0$ be an integer and let $\beta=4n^2+4n$.
    Then the following are equivalent.
    \begin{itemize}
        \item The conflict graph $H$ has an acyclic subset of size at least $k$.
        \item The compatible flip distance $\dist_{\rm comp}(\beta \cdot T, \beta \cdot T') \leq (\beta+1)(2|H| - k) + |H| + 4(n-|H|)$.
        \item The rotation flip distance $\dist_{\rm rot}(\beta \cdot T, \beta \cdot T') \leq (\beta+1)(2|H| - k) + |H| + 4(n-|H|)$.
    \end{itemize}
\end{theorem}
\begin{proof}
    Assume first that $\las(H) < k$.
    Then by \cref{lem:lower}
    \begin{align*}
        \dist(\beta\cdot T, \beta\cdot T') & \geq (\beta-2n)(2|H| - \las(H)) \geq (\beta+1-(2n+1))(2|H|-k+1)\\
        &= (\beta+1)(2|H|-k) + (\beta+1) - (2n+1)(2|H|-k+1) \\
        &> (\beta+1)(2|H|-k) + |H| + 4\cdot(n-|H|).
    \end{align*}
    For the last inequality observe that $|H| < n$ and hence $(2n+1)(2|H|-k+1) \leq (2n+1)(2n-1) \leq 4n^2$, which together with $\beta = 4n^2 + 4n > 4n^2 + 4(n-|H|) + |H|$ gives the desired.
    
    Conversely, assume $H$ has an acyclic subset $S$ of size $k$.
    Each gap $g \in V(H)$ corresponds to a set $B(g)$ of $\beta+1$ gaps in $H^\beta = H(\beta\cdot T, \beta \cdot T')$.
    Then $S^\beta = \bigcup_{g \in S} B(g)$ is an acyclic set in $H^\beta$ of size $(\beta+1)k$. 
    We execute the following flip sequence from $\beta\cdot T$ to $\beta \cdot T'$.
    \begin{description}
        \item[(1)] Flip each edge $e \in T$ that does not correspond to a gap $g \in S^\beta$ to its gap $g$.
        For a near edge $e$ this is a rotation, since $e$ shares a vertex with its gap $g$.
        For a wide edge $e$ we do this in two rotations.
        The first rotation keeps the right endpoint of $e$ fixed and moves the left endpoint of $e$ to the left endpoint of the gap $g$.
        The second rotation then moves the right endpoint to the gap.
        
        \item[(2)] Consider the gaps $g \in S$ in a topological ordering $t$.
        If $g$ is above (respectively below), flip all edges in $\beta \cdot T$ corresponding to gaps in $B(g)$ directly to their target location in $\beta\cdot T'$ in order from shortest to longest edge (respectively longest to shortest edge).
        All these flips are rotations.
        If $g$ is crossing, we use the $\beta+2$ rotations (for the $\beta +1$ gaps in $B(g)$) as given by \cref{lemma:rotbound}.
        
        \item[(3)] Flip each edge $e$ that is flipped to its gap $g$ in (1) to its target location in $\beta \cdot T'$ in at most two rotations.
    \end{description}
    
    In step (1) we spend at most four rotations for every edge pair that is not a near-near pair.
    The maximum is attained for wide-wide pairs $(e,e')$ where we need two rotations to get $e$ to the convex hull and two to rotate from there to $e'$.
    In total, these are at most $4\cdot(n-|H|)$ flips for all non-near-near pairs.
    For each of the $(\beta+1)(|H|-k)$ near-near pairs in step (1), we spend two flips.
    
    In step (2) we spend one flip per gap in $S^\beta$, that is, $|S^\beta| = (\beta+1)k$ flips.
    But for each crossing gap $g \in S$ we spend one extra flip, giving at most $|S| \leq |H|$ additional flips, and thus at most $(\beta+1)k + |H|$ flips in step (2) in total.
        
    In consequence, we obtain an upper bound on the length of the rotation (and thus compatible) flip distance of 
    \[
        (\beta+1)(2|H|-k) + |H| + 4(n-|H|),
    \]
    as desired.
\end{proof}

%%%%%%%%%%%%%%%%%%%%%
%%                 %%
%%  STACKED TREES  %%
%%                 %%
%%%%%%%%%%%%%%%%%%%%%
\section{Stacked trees}\label{sec:stacked}
Let $T$ and $T'$ be two non-crossing spanning trees on a set of $n$ linearly ordered points.
According to \cref{thm:bigtheorem} and \cref{ac:ABC_SODA}, for the flip distance between $T$ and $T'$ we have $\dist(T,T') 
%\leq \max\{\frac{3}{2}, 2-\frac{\las(H)}{\lvert V(H) \rvert}\}(n-1)\} 
\leq 
\nicefrac53(n-2)$, where $H = H(T,T')$ is the associated conflict graph.
In this section, we improve this upper bound 
%on $\dist(T,T')$
for the special case that in one of the trees, say $T$, the (relevant) edges can be partitioned into independent ``stacks'' of nested edges.
In fact, any two edges $e_1,e_2$ in a non-crossing tree $T$ in a linear representation are either nested ($e_1$ covers $e_2$ or vice versa) or disjoint (the subset of gaps covered by $e_1$ and $e_2$ are disjoint).
We call a non-crossing tree $T$ in a linear representation \emph{stacked} if the edge set of $T$ can be partitioned into sets $S_1,\ldots,S_k$ such that
\begin{romanenumerate}
    \item no edge in $S_i$ covers an edge in $S_j$ for $i\neq j$, and\label{item:stacked-1}
    \item each $S_i$ is totally ordered by the covering relation.\label{item:stacked-2}
\end{romanenumerate}
In other words, $T$ is stacked if and only if there are no three edges $e_1,e_2,e_3$ in $T$ such that $e_1$ covers both $e_2$ and $e_3$, while $e_2$ and $e_3$ are disjoint.

A stacked tree can be seen as a natural generalization of a so-called separated caterpillar, that is, a tree with at most two stacks.
In fact, it is known that if $T$ is a separated caterpillar, then $\dist(T,T') \leq \nicefrac32(n-2)$~\cite[Section~6, arXiv-Version]{bjerkevik2024flippingnoncrossingspanningtrees}.
On the other hand, the best previously known lower bound example in~\cite{bjerkevik2024flippingnoncrossingspanningtrees}, depicted in \cref{fig:old}, is a pair $T,T'$ with $\dist(T,T') \geq \frac{14}{9}n - O(1)$ where one tree $T$ is stacked.
In this section, we show that this lower bound of $\frac{14}{9}n$ is essentially tight whenever $T$ is stacked, by this proving  \cref{thm:stackedTrees}.

Throughout this section, we refer to a non-crossing spanning tree with a linear representation simply as a tree.
Let $T,T'$ be a pair of trees.
We say that $T$ is \emph{stacked with respect to $T'$} if the subset $N$ of edges in $T$ that are in near-near pairs of $(T,T')$ admits a partition into sets $S_1,\ldots,S_k$ satisfying properties \eqref{item:stacked-1} and \eqref{item:stacked-2} in the definition of a stacked tree.
Note that if $T$ is stacked then $T$ is stacked with respect to any other tree $T'$.
We refer to the edge sets $S_i$ of $T$ as \emph{stacks of $T$ (with respect to $T'$)}.

\begin{figure}[htb]
    \centering
    \includegraphics[page=1]{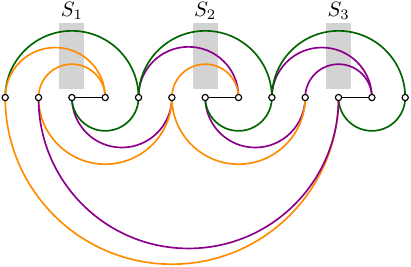}
    \caption{
        The best previously known lower bound example from~\cite{bjerkevik2024flippingnoncrossingspanningtrees}. 
        The top tree is a stacked tree with three stacks.
        The bottom tree is neither stacked, nor stacked with respect to the top tree.
    }
    \label{fig:old}
\end{figure}

%%%%%%%%%%%%%%%%%%%%%%
%%   THREE STACKS   %%
%%%%%%%%%%%%%%%%%%%%%%
\subsection{The Special Case: Three stacks}
Fix trees $T$ and $T'$. 
We first assume that $T$ has exactly three stacks $S_1$, $S_2$ and $S_3$ with respect to $T'$.
Let $\mathcal{P}_N$ be the set of near-near pairs of $(T,T')$.
We partition $\mathcal{P}_N$ into $9$ different sets. %$A_1,A_2,A_3,B_{1,2},B_{1,3},B_{2,1},B_{2,3},B_{3,1},B_{3,2}$.
In particular, pairs containing edges from $S_i$ are partitioned into three sets $A_i, B_{i,j}, B_{i,k}$ where $\{i,j,k\} = \{1,2,3\}$.
A pair $(e,e')\in\mathcal{P}_N$ with $e\in S_i$ is put into a set according to the following rules.
\begin{itemize}
    \item If $(e,e')$ is an above pair, then $(e,e') \in A_i$.
    \item If $(e,e')$ is a below pair or a crossing pair, and $e'$ crosses an edge in a different stack $S_j$, $i\neq j$, then $(e,e') \in B_{i,j}$.
    \item If  $(e,e')$ is a below pair or a crossing pair, and $e'$ does not cross any edge in any other stack, we assign $(e,e')$ arbitrarily to one of the two sets $B_{i,j}$ with $j \neq i$.
\end{itemize}
For this to be well-defined, we need to verify that $(e,e')$ cannot belong to both $B_{i,j}$ and $B_{i,j'}$ for $j\neq j'$.
The edge $e'$ covers two gaps adjacent to its endpoints.
One of these gaps is also covered by~$e$, since $(e,e')$ is a near-near pair.
Now, if $e'$ crosses an edge $f\in S_j$ with $j \neq i$, then the gap at the other end of $e$ is covered by $f$.
Hence every edge of $T$ that crosses $e'$ belongs to the same stack, and our partition is well-defined.

\begin{lemma}\label{lem:threestack}
    For all choices of $x,y,z$ such that $\{x,y,z\} = \{1,2,3\}$, both $H_x: = A_y \cup A_z \cup B_{x,y}\cup B_{x,z}$ and $H_{x,y} := A_z \cup B_{y,z}\cup B_{x,y} \cup B_{x,z}$ are acyclic. 
\end{lemma}
\begin{proof}
    By  \cref{ac:ABC_SODA},  $A_1\cup A_2\cup A_3$  is acyclic because it only has above gaps. 
    % \linda{shouldn't we ref to \cref{ac:ABC_SODA}}\cite[Lemma 3.2]{bjerkevik2024flippingnoncrossingspanningtrees}, $A_1\cup A_2\cup A_3$ is non-crossing\linda{acyclic???}, since it only has above gaps. 
    Hence, any subset of $A_1\cup A_2\cup A_3$ is acyclic.
    We show that $B_{i,j} \cup B_{i,k}$ for any $i,j,k$ such that $\{i,j,k\} = \{1,2,3\}$
    %, \sout{with two sets $B_{i,j}$ and $B_{i,k}$ of edges from the same stack}, 
    is acyclic. 
    For this, 
    pick the gap $g$ corresponding to the pair $(e,e')$ in $B_{i,j} \cup B_{i,k}$ where $e$ is of maximal length among all such pairs. 
    It suffices to show that in the conflict graph $H = H(T,T')$ there is no edge $h\to g$, where $h$ is a gap associated to another pair $(f,f')$ in $B_{i,j} \cup B_{i,k}$. 
    Note that since $e$ is of maximal length and $f$ is in the same stack as $e$ (namely, $S_i$), $e$ covers $f$. 
    We rule out each type of conflict between $h$ and $g$:
\begin{description}
    \item [type 1] If $(e,e')$ is a below pair, $e'$ covers $e$, which covers $f$, so $e'$ and $f$ do not cross.
    \item[type 2] Since $e$ covers $f$ and $e$ the minimal edge in $T$ covering $g$, $f$ does not cover $g$.
    \item [type 3] Since $(e,e')$ is either below or crossing, $e$ does not cover $e'$, and thus, neither does $f$.
\end{description}
%    Type 2: Since $e$ covers $f$ and $e$ the minimal edge in $T$ covering $g$, $f$ does not cover $g$.

 %   Type 1: If $(e,e')$ is a below pair, $e'$ covers $e$, which covers $f$, so $e'$ and $f$ do not cross.
   % If $(e,e')$ is a crossing pair, $e'$ and $f$ do not cross, since $f$ does not cover $g$.

  %  Type 3: Since $(e,e')$ is either below or crossing, $e$ does not cover $e'$, and thus, neither does $f$.

    With this, we have shown that $g$ has no incoming conflict edge from any gap in $B_{i,j} \cup B_{i,k}$.
    Removing $g$ from the set and repeating the procedure proves acyclicity.
    By these acyclicity observations, we can choose topological orderings $t_i$, $t_{i,j}$ and $t_{i,j,k}$ for $A_i$, $B_{i,j}$ and $B_{i,j} \cup B_{i,k}$, respectively, for all values of $i,j,k$.

    Let $(e,e')$ be the edge pair of a gap $g$ and let $(f,f')$ be the edge pair of a gap $h$, where $g \in A_i \cup B_{i,j} \cup B_{i,k}$ and $h\in A_j \cup B_{j,k}$, and $i$, $j$ and $k$ are all different.
    Since $e$ neither covers nor crosses $f'$, and does not cover $h$, there is no edge from $g$ to $h$.
    Using this observation for different values of $i$, $j$ and $k$, we get that the concatenation $t_y t_z t_{x,y,z}$ is a topological ordering of $H_x$, and $t_z t_{y,z} t_{x,y,z}$ is a topological ordering of $H_{x,y}$.
    Thus, $H_x$ and $H_{x,y}$ are acyclic.
\end{proof}

We are now ready to obtain the following upper bound on the flip distance $\dist(T,T')$ of two non-crossing trees if one, say $T$, is a stacked tree with three stacks.

\begin{theorem}\label{thm:condupper}
    Let $T$, $T'$ be non-crossing trees on $n\geq 3$ points in convex position.
    Let $T$ be a stacked tree with three stacks.
    Then $\dist(T,T') \leq \nicefrac{14}{9}(n-1)$.
\end{theorem}
\begin{proof}
    Let $H = H(T,T')$ be the conflict graph.
    If $V(H)$ is empty, then $\dist(T,T') \leq \nicefrac32(n-1)$ by \cref{thm:bigtheorem}\eqref{item:big-UB}.
    So assume that $V(H)$ is non-empty.
    Then each element of $V(H)$ is contained in four of the nine acyclic sets $H_1,H_2,H_3$, $H_{1,2},H_{2,1},H_{1,3},H_{3,1},H_{2,3},H_{3,2}$ from \cref{lem:threestack}.
    Thus the average size of these sets is $\nicefrac 19\Sigma_{i=1}^9 \lvert H_i \rvert = \nicefrac 49\lvert V(H)\rvert$ and, consequently, the largest acyclic set is of size at least~$\nicefrac{4}{9}\lvert V(H)\rvert$.
    Plugging this into \cref{thm:bigtheorem}\eqref{item:big-UB}, we obtain $\dist(T,T') \leq \nicefrac{14}{9}(n-1)$.
\end{proof}

\begin{remark}
    For the sake of self-containment, we explain how a flip sequence from $T$ to $T'$ for \cref{thm:condupper} looks like.
    Let $A \subseteq V(H)$ be an acyclic subset of size $|A| \geq \nicefrac{4}{9}|V(H)|$.
    First, flip all edges of $T$ corresponding to pairs that are not in $A$ to their gap on the convex hull in arbitrary order. 
    Then, flip the edges of $T$ corresponding to gaps in $A$ directly to their target location in $T'$ according to a topological ordering of $A$. 
    Lastly, flip all edges that were flipped to their gaps in the first phase to their target location in $T'$, again in arbitrary order.
\end{remark}

%%%%%%%%%%%%%%%%%%%%%%
%%    MANY STACKS   %%
%%%%%%%%%%%%%%%%%%%%%%
\subsection{Trees with arbitrarily many stacks}
We now consider a pair of trees $T$ and $T'$ where $T$ is a stacked tree with any number of stacks.
Let $S_1,\ldots,S_k$ denote the stacks of $T$ with respect to $T'$. 
Let $G$ be the graph with $S_1,\ldots,S_k$ as its vertices and an edge between $S_i$ and $S_j$ with $i\neq j$ if there is a near-near pair $(e,e')$ such that $e$ is in $S_i$ and $e'$ crosses an edge in $S_j$.
As $T'$ is non-crossing, $G$ is a subset of a triangulation of a convex polygon with $k$ vertices.
Therefore $G$ is admits a proper $3$-coloring of its vertices.
Fix any such coloring and let the colors be denoted by $1$, $2$ and $3$.

\begin{figure}[ht]
    \centering
    \begin{subfigure}{.65\textwidth}
        \centering
        \includegraphics[page=3]{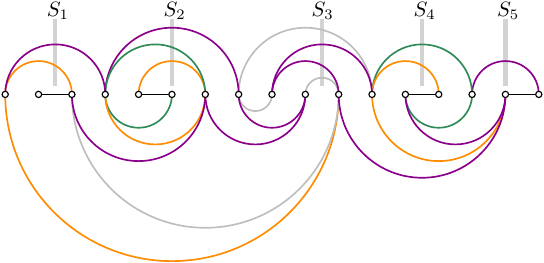}
        \subcaption{}
        \label{fig:StacksA}
    \end{subfigure}%
    \hfil
    \begin{subfigure}{.3\textwidth}
        \centering
        \includegraphics[page=2]{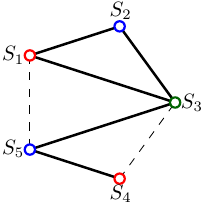}
        \subcaption{}
        \label{fig:StacksB}
    \end{subfigure}
    \caption{
        (a) A pair of trees where one tree is stacked and (b) the resulting graph $G$ with a proper $3$-coloring. 
    }
    \label{fig:Stacks}
\end{figure}

We partition $\mathcal{P}_N$ into $9$ different sets $A_1,A_2,A_3,B_{1,2},B_{1,3},B_{2,1},B_{2,3},B_{3,1},B_{3,2}$.
A pair $(e,e')\in\mathcal{P}_N$ where  $e$ belongs to stack with color $i$ is put into a set according to the following rules.
\begin{itemize}
    \item If $(e,e')$ is an above pair, then $(e,e') \in A_i$.
    \item If  $(e,e')$ is a below pair or a crossing pair, and $e'$ crosses an edge in a different stack of color $j$, $i\neq j$, then $(e,e') \in B_{i,j}$.
    \item If  $(e,e')$ is a below pair or a crossing pair, and $e'$ does not cross any edge in any other stack, we assign $(e,e')$ arbitrarily to one of the two sets $B_{i,j}$ with $j \neq i$.
\end{itemize}

\cref{lem:allstack}, which is a slight extension of \cref{lem:threestack}, can be used to prove \cref{thm:stackedupper} in a similar way as \cref{thm:condupper} is derived from \cref{lem:threestack}.

\begin{lemma}\label{lem:allstack}
    Let $\{x,y,z\} = \{1,2,3\}$.
    Then $H_x = A_y \cup A_z \cup B_{x,y}\cup B_{x,z}$ and $H_{x,y} = A_z \cup B_{y,z}\cup B_{x,y} \cup B_{x,z}$ are acyclic.
\end{lemma}
\begin{proof}
    The proof of \cref{lem:threestack} goes through verbatim, except that we need to show that $B_{i,j} \cup B_{i,j'}$ is acyclic in this altered setting.
    We again pick a gap $g$ given by $(e,e')$ in $B_{i,j} \cup B_{i,j'}$ where $e$ is of maximal length, and want to show that there is no edge $h\to g$, where $h$ is a gap associated to another pair $(f,f')$ in $B_{i,j} \cup B_{i,j'}$.
    In \cref{lem:threestack}, we dealt with the case where $e$ and $f$ are in the same stack.
    We now handle the case where they are in different stacks.
    We immediately have that $f$ does not cover $g$ or $e'$, so there is no edge $h\to g$ of type 2 or 3.
    The stacks of $g$ and $h$ have the same color, so by definition of the coloring, $e'$ does not cross $f$, which means that there is no edge $h\to g$ of type 1, either.
    We conclude that $B_{i,j} \cup B_{i,j'}$ is acyclic for the same reasons as in the proof of \cref{lem:threestack}.
\end{proof}

\begin{theorem}\label{thm:stackedupper}
    Let $T$, $T'$ be non-crossing trees on $n\geq 3$ points in convex position.
    Let $T$ be stacked with respect to $T'$.
    Then $\dist(T,T') \leq \nicefrac{14}{9}(n-1)$.
\end{theorem}

We remind the reader that if $T$ is stacked, then it is stacked with respect to any $T'$.
Thus, the same statement without ``with respect to $T'$'' is true a fortiori.

% \newpage

%%%%%%%%%%%%%%%%%%%%%
%%                 %%
%%   LOWER BOUND   %%
%%                 %%
%%%%%%%%%%%%%%%%%%%%%
\section{A new lower bound}\label{sec:newlowerbound}
In this section, we improve the lower bound on the diameter $\diam(\mathcal{F}_n)$ of the flip graph of non-crossing spanning trees on $n$ points in convex position.
To this end, we iteratively construct for each integer $\ell \geq 1$ a pair of trees $(T^\ell,\tilde T^\ell)$, and show that the corresponding conflict graph $H(T^\ell,\tilde T^\ell)$ has $7\ell$ vertices and the largest acyclic subsets of size at most $3\ell+1$.
The latter divided by the former tends to~$\nicefrac 37$ as $\ell$ tends to infinity, which allows us to use \cref{thm:bigtheorem}\eqref{item:big-LB} to prove that $\diam(\mathcal{F}_n) \geq \nicefrac{11}{7}\cdot n - o(n)$.
            
We start with the trees $T^1$ and $\tilde T^1$ illustrated in \cref{fig:first}.
These are trees on vertices $v_1, \dots, v_8, x, y, v_9$ labeled from left to right along the spine and have short edges $v_2v_3$, $v_6v_7$ and $xy$ as well as edges $e_i$ and $\tilde e_i$ for $i=1,\dots, 7$, respectively.
For completeness, we give a formal definition of $T^1$ and $\tilde T^1$.
Both trees have short edges $v_2v_3$, $v_6v_7$ and $xy$. 
Tree $T^1$ has near edges $e_1 = v_1v_5, e_2 = v_2v_4, e_3 = v_2v_5, e_4 = v_5x, e_5 = v_6v_8, e_6 = v_6x,  e_7 = xv_9$.
Tree $\tilde T^1$ has near edges $\tilde e_1 = v_1v_3, \tilde e_2 = v_3y, \tilde e_3 = v_4v_7, \tilde e_4 = v_5v_7, \tilde e_5 = v_4v_8, \tilde e_6 = v_8y,  \tilde e_7 = v_3v_9$.
The near-near pairs are $(e_i, \tilde e_i)$ with corresponding gap $g_i$ for all $i$.
Define $e_j^1=e_j$, $\tilde e_j^1=\tilde e_j$ and $g_j^1=g_j$ for all $1\leq j\leq 7$.

\begin{figure}[ht]
    \centering
    \begin{subfigure}{.45\textwidth}
        \centering
        \includegraphics[page=1]{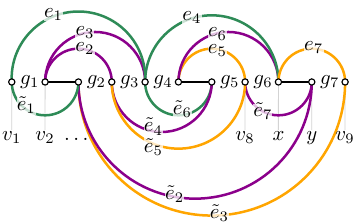}
        \subcaption{}
    \end{subfigure}\hfil\begin{subfigure}{.55\textwidth}
        \centering
        \includegraphics[page=2]{figures/Low_constructionN}
        \subcaption{}
    \end{subfigure}
    \caption{
        Illustration of $T^1$ and $\tilde T^1$ and the double conflicts in the corresponding conflict graph.
    }
    \label{fig:first}
\end{figure}

The trees $T^1$ and $\tilde T^1$ are shown in \cref{fig:first}.
The double conflicts in their conflict graph $H(T^1,\tilde T^1)$ form a path as illustrated in the figure.
Clearly, $H(T^1,\tilde T^1)$ has $7$ vertices and the largest acyclic subsets of size at most $4$.

Next, we construct $T^\ell$ and $\tilde T^\ell$ inductively for $\ell\geq 2$.
Assume that $T^{\ell-1}$ and $\tilde T^{\ell-1}$ are given. 
Add vertices $v_1^\ell, \dots, v_5^\ell$ from left to right to the left of $v_1^{\ell-1}$ and let $v_6^\ell=v_2^{\ell-1}$.
Further, add a vertex $v_7^\ell$ between $v_3^{\ell-1}$ and $v_4^{\ell-1}$, a vertex $v_8^\ell$ between $v_8^{\ell-1}$ and $x$, and a vertex $v_9^\ell$ to the right of $v_9^{\ell-1}$.
For an illustration consider \cref{fig:second}.

\begin{figure}[ht]
    \centering
    \includegraphics[page=3]{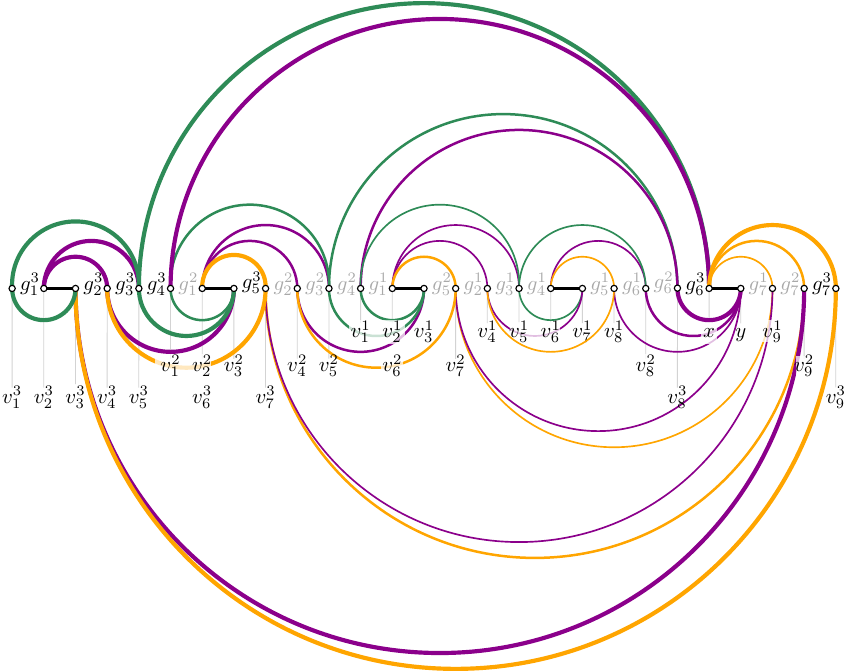}
    \caption{
        Illustration for the iterative construction of $T^\ell$ and $\tilde T^\ell$ from $T^{\ell-1}$ and $\tilde T^{\ell-1}$ for $\ell = 3$.
    }
    \label{fig:second}
\end{figure}
            
The edges in $T^{\ell-1}$ and $\tilde T^{\ell-1}$ are also present in $T^\ell$ and $\tilde T^\ell$, respectively, except that the right endpoint of $e_4^{\ell-1}$ and $e_6^{\ell-1}$ (which was $x$) is changed to $v_8^\ell$, and the left endpoint of $\tilde e_2^{\ell-1}$ and $\tilde e_7^{\ell-1}$ (which was $v_3^{\ell-1}$) is changed to $v_7^\ell$.
Correspondingly, we change $g_2^{\ell-1}$ to $v_7^\ell v_4^{\ell-1}$ and $g_6^{\ell-1}$ to $v_8^{\ell-1} v_8^\ell$.
We put a short edge $v_2^\ell v_3^\ell$ in both $T^\ell$ and $\tilde T^\ell$.
Lastly, we add the following edges to $T^\ell$ and $\tilde T^\ell$
\begin{align*}
    &e_1^\ell = v_1^\ell v_5^\ell, \hspace{.4cm} e_2^\ell = v_2^\ell v_4^\ell, \hspace{.4cm} e_3^\ell = v_2^\ell v_5^\ell, \hspace{.4cm} e_4^\ell = v_5^\ell x, \hspace{.4cm} e_5^\ell = v_6^\ell v_7^\ell, \hspace{.4cm} e_6^\ell = v_1^{\ell-1} x, \hspace{.4cm} e_7^\ell = xv_9^\ell,\\
    &\tilde e_1^\ell = v_1^\ell v_3^\ell, \hspace{.4cm} \tilde e_2^\ell = v_3^\ell v_9^{\ell-1}, \hspace{.1cm} \tilde e_3^\ell = v_4^\ell v_3^{\ell-1}, \hspace{.1cm} \tilde e_4^\ell = v_5^\ell v_3^{\ell-1}, \hspace{.1cm} \tilde e_5^\ell = v_4^\ell v_7^\ell, \hspace{.4cm} \tilde e_6^\ell = v_8^\ell y, \hspace{.6cm} \tilde e_7^\ell = v_3^\ell v_9^\ell,
\end{align*}
with corresponding gaps
\begin{align*}
    &g_1^\ell = v_1^\ell v_2^\ell, \hspace{.4cm} g_2^\ell = v_3^\ell v_4^\ell, \hspace{.4cm} g_3^\ell = v_4^\ell v_5^\ell, \hspace{.4cm} g_4^\ell = v_5^\ell v_1^{\ell-1}, \hspace{.4cm} g_5^\ell = v_3^{\ell-1} v_7^\ell, \hspace{.4cm} g_6^\ell = v_8^\ell x, \hspace{.4cm} g_7^\ell = v_9^{\ell-1} v_9^\ell.
\end{align*}
            
We show the following properties for $T^\ell$ and $\tilde T^\ell$.

\begin{lemma}\label{lem:lower_bound_construction}
    The following are true for $T^\ell$ and $\tilde T^\ell$:
    \begin{enumerate}[(1)]
        \item $T^\ell$ and $\tilde T^\ell$ are both non-crossing spanning trees,\label{item:new-LB-1}
        \item their near-near pairs are $(e_i^j,\tilde e_i^j)$, with associated gaps $g_i^j$, for $1\leq i\leq 7$ and $1\leq j\leq \ell$,\label{item:new-LB-2}
        \item for $1\leq j\leq \ell$, the conflict graph has bidirected edges $g_7^j \leftrightarrow g_2^j \leftrightarrow g_1^j \leftrightarrow g_3^j \leftrightarrow g_5^j \leftrightarrow g_4^j \leftrightarrow g_6^j$,\label{item:new-LB-3}
        \item for $1\leq j< j' \leq \ell$, the conflict graph has a bidirected edge $g_6^j \leftrightarrow g_7^{j'}$.\label{item:new-LB-4}
    \end{enumerate}
\end{lemma}

\begin{figure}[h]
    \centering
    \includegraphics[page=4]{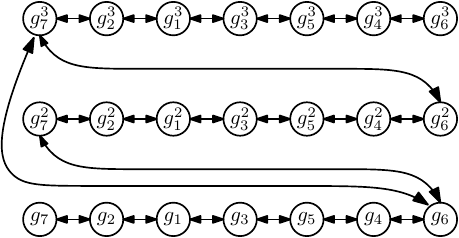}
    \caption{
        Double conflicts in the conflict graph $H(T^\ell,\tilde T^\ell)$ for $l=3$.
    }
    \label{fig:conflictGraph}
\end{figure}

\begin{proof}
    Checking that \eqref{item:new-LB-1}--\eqref{item:new-LB-3} hold for $\ell=1$ we leave to the reader.
    For $\ell \geq 2$ we proceed by induction.

    \begin{description}
        \item[\eqref{item:new-LB-1}:] 
            First, we show that $T^\ell$ is a non-crossing spanning tree.
            In fact, the edges $xy, e_7^1,\dots, e_7^\ell$ form a spanning tree $S$ on the vertices $x,y, v_9^1,\dots, v_9^\ell$, which is non-crossing, since they all share the vertex $x$.
            We claim that the rest of the edges of $T^\ell$ form a non-crossing spanning tree on the vertices to the left of $x$, including $x$.
            This is true for $i=1$, so we can take this as part of the induction hypothesis.
            The edges $v_2^\ell v_3^\ell, e_1^1,\dots, e_4^\ell, e_6^\ell$ form a non-crossing spanning tree $S'$ on the vertices $v_1^\ell,\dots, v_5^\ell, v_1^{\ell-1}, x$.
            Moreover, by the induction hypothesis, the edges we have not yet accounted for except $e_5$ form a spanning tree $S''$ on the vertices from $v_1^{\ell-1}$ to $v_8^\ell$ (taking into account the shifting of the endpoints of $e_4^\ell$ and $e_6^\ell$) except $v_7^\ell$.
            Lastly, $e_5^\ell$, connects this tree to $v_7^\ell$.
            We have that $S'$ shares one vertex with $e_5^\ell$ and one with $S''$, which shares one vertex with $S$, so the union is a spanning tree on the full set of vertices.

            There are no distinct vertices $a,b,c,d$ in order from left to right or from right to left with $a,c\in S$ and $b,d\in S'$, so no edge in $S$ crosses an edge in $S'$.
            Similar arguments exclude all crossings except between $e_5^\ell$ and an edge in $S''$.
            But the only edge in $T^\ell$ with an endpoint strictly between the endpoints of $e_5^\ell$ is short, so $e_5^\ell$ does not cross any edge.
            Thus, $T^\ell$ is a non-crossing spanning tree.
            
            Next, we show that $\tilde T^\ell$ is a non-crossing spanning tree.
            The only edges incident to the two leftmost vertices $v_1^\ell$ and $v_2^\ell$ are $\tilde e_1^\ell$ and $v_2^\ell v_3^\ell$.
            Thus, it suffices to show that what remains after removing $v_1^\ell$, $v_2^\ell$, $\tilde e_1^\ell$ and $v_2^\ell v_3^\ell$ forms a non-crossing spanning tree, and we can take this as part of our induction hypothesis.
            The edges $\tilde e_2^\ell, \dots, \tilde e_5^\ell, \tilde e_7^\ell$ form a non-crossing spanning forest whose connected components have vertex sets $\{v_4^\ell, v_5^\ell, v_3^{\ell-1}, v_7^\ell\}$ and $\{v_3^\ell, v_9^{\ell-1}, v_9^\ell\}$.
            By the induction hypothesis, the rest of the edges except $\tilde e_1^{\ell-1}$, $v_2^{\ell-1}v_3^{\ell-1}$ and $\tilde e_6^\ell$ form a non-crossing spanning tree on the vertices from $v_7^\ell$ to $v_9^{\ell-1}$ except $v_8^\ell$.
            (Note that we shifted the left endpoint of $\tilde e_2^{\ell-1}$ and $\tilde e_7^{\ell-1}$ from $v_3^{\ell-1}$ to $v_7^\ell$.)
            Taking the union of the tree and the forest, we get a spanning tree on the union of the vertex sets.
            This tree is non-crossing for reasons similar to our argument above that no edge in $S$ crosses any edge in $S'$.
            Finally, we add $v_1^{\ell-1}$ and $v_2^{\ell-1}$ with $\tilde e_1^{\ell-1}$ and $v_2^{\ell-1}v_3^{\ell-1}$, and $v_8^\ell$ with $\tilde e_6^\ell$ to obtain a non-crossing spanning tree on the full set of vertices.

        \item[\eqref{item:new-LB-2}:]
            The gaps $g_1^\ell, \dots, g_4^\ell, g_6^\ell, g_7^\ell$ are only covered by edges of the form $e_i^\ell$ in $T^\ell$, and the gaps $g_1^\ell, \dots, g_5^\ell, g_7^\ell$ are only covered by edges of the form $\tilde e_i^\ell$ in $\tilde T^\ell$.
            Moreover, the shortest edge in $T^\ell$ covering $g_5^\ell$ is $e_5^\ell$, and the shortest edge in $\tilde T^\ell$ covering $g_6^\ell$ is $\tilde e_6^\ell$.
            Thus, all the edges $e_1^\ell, e_7^\ell$ and $\tilde e_1^\ell, \tilde e_7^\ell$ are associated to the gaps $g_1^\ell, g_7^\ell$, and one can check that $e_i^\ell$ and $\tilde e_i^\ell$ are associated to $g_i^\ell$ for all $1\leq i\leq 7$.
            By the induction hypothesis, the same holds for $j<\ell$ in $T^{\ell-1}$ and $\tilde T^{\ell-1}$.
            Since we have not changed any covering relations between edges and gaps $e_i^j$, $g_i^j$, or between edges and gaps $\tilde e_i^j$, $g_i^j$ for $1\leq i\leq 7$ and $1\leq j<\ell$, even with our changes of endpoints, we preserve all the pairings of edges and gaps that we need to prove \eqref{item:new-LB-2}.

        \item[\eqref{item:new-LB-3} and \eqref{item:new-LB-4}:]
            First note that all conflicts mentioned in \eqref{item:new-LB-3} and \eqref{item:new-LB-4} are of type 1.
            Thus, it suffices to check that for every edge $g_i^j \leftrightarrow g_{i'}^{j'}$ that we claim, $e_i^j$ crosses $\tilde e_{i'}^{j'}$, and $e_{i'}^{j'}$ crosses $\tilde e_i^j$.
            The changing of endpoints in our construction does not destroy any crossings, so \eqref{item:new-LB-3} can be checked straightforwardly for $j=\ell$ and similarly, \eqref{item:new-LB-4} can be checked for $j'=\ell$.\qedhere
    \end{description}
\end{proof}
		
\begin{lemma}
    \label{lem:lower_bound_ac}
    The conflict graph $H(T^\ell, \tilde T^\ell)$ has $7\ell$ vertices and its largest acyclic sets have size $3\ell+1$.
\end{lemma}
\begin{proof}
    That the conflict graph has $7\ell$ vertices follows immediately from \cref{lem:lower_bound_construction}\eqref{item:new-LB-2}.
    Observe that from every bidirected path of length seven given in \cref{lem:lower_bound_construction}\eqref{item:new-LB-3}, we can pick at most four gaps for the acyclic set.
    Furthermore, the only way to pick four is to pick the gaps $g_7^j$, $g_1^j$, $g_5^j$ and $g_6^j$, see also \cref{fig:conflictGraph}.
    Now assume the cardinality of an acyclic set exceeds~$3\ell+1$.
    Then, by the pigeonhole principle, there exist at least two paths of length seven from which there are four pairs in the acyclic set.
    In particular, the set contains $g_6^j$ and $g_7^{j'}$ for some $j<j'$ which are in double conflict by \cref{lem:lower_bound_construction}\eqref{item:new-LB-4}, contradicting acyclicity.
\end{proof}

Together, \cref{lem:lower_bound_ac,thm:bigtheorem} imply the following.

\begin{theorem}
    As a function of $n$, $\diam(\mathcal{F}_n) \geq \nicefrac{11}{7}\cdot n-o(n)$.
\end{theorem}
\begin{proof}
    By \cref{lem:lower_bound_ac} and \cref{thm:bigtheorem}\eqref{item:big-LB}, for any fixed $\ell$, we have
    \[
        \diam(\mathcal{F}_{n}) \geq \Big(2-\frac{3\ell+1}{7\ell}\Big)n - c = \Big(\frac{11}{7}-\frac{1}{7\ell}\Big)n - c
    \]
    for some $c$ depending on $\ell$.
    As a consequence, for any $\epsilon>0$, $\diam(\mathcal{F}_{n}) \geq \big(\frac{11}{7}-\epsilon\big)n$ for sufficiently large $n$.
    In fact, choose $\ell$ such that $\frac{1}{7\ell} < \epsilon$ and $n$ large enough that $(\epsilon-\frac{1}{7\ell})n\geq c$, and calculate
    \[
        \diam(\mathcal{F}_{n}) \geq \Big(\frac{11}{7}-\frac{1}{7\ell}\Big)n - c = \Big(\frac{11}{7}-\epsilon\Big)n + \Big(\epsilon - \frac{1}{7\ell}\Big)n - c \geq \Big(\frac{11}{7}-\epsilon\Big)n.
    \]
    
    It follows that $\diam(\mathcal{F}_{n}) \geq \frac{11}{7}n-o(n)$.
\end{proof}

%%%%%%%%%%%%%%%%%%%%%
%%                 %%
%%   CONCLUSIONS   %%
%%                 %%
%%%%%%%%%%%%%%%%%%%%%
\section{Conclusion}\label{sec:conclusion}
We provide new insights on the flip distance of non-crossing spanning trees.
Firstly, we show that computing the flip distance of two trees on points in convex position is \NP-hard. 
To the best of our knowledge this constitutes the first hardness result of a reconfiguration problem of plane spanning graphs for convex point sets, see also \cref{table:2}. 
The hardness also remains true when restricting to flips that are compatible or rotations.
This insight disproves the belief that computing the flip distance lies in \P{} whenever the happy edge property holds.

Secondly, our new lower bound on the diameter of the flip graph of non-crossing spanning trees on convex point sets improves the current best bounds in several related settings, including points in general position, compatible and rotation flips for both convex and general points sets, as well as slides for convex point sets.

Interesting future directions include determining the exact diameters for flip graphs on point sets in convex as well as general position.
We are also hopeful that our techniques for NP-hardness are fruitful in other settings. Very exciting is a recently announced breakthrough by Dorfer~\cite{dorfer2026flip} concerning the complexity of computing the flip distance of triangulations on convex point sets.
\bibliography{citationnew,lit}
	
\end{document}